\documentclass[preprint,12pt]{elsarticle}

\usepackage{amssymb}
\usepackage{graphicx} % Modified by Arthur
\usepackage{epstopdf}
\usepackage{booktabs}		

\usepackage{amsmath}
\usepackage{amsfonts}

	\usepackage{amsthm}

\usepackage{url}

\theoremstyle{plain}
\newtheorem{thm}{Theorem}
\newtheorem{lem}{Lemma}

\theoremstyle{definition}
\newtheorem{defin}{Definition}
\newtheorem{problem}{Problem}

\theoremstyle{remark}
\newtheorem{remark}{Remark}

\usepackage{stmaryrd} % for double bracket
\SetSymbolFont{stmry}{bold}{U}{stmry}{m}{n}

\usepackage{diagbox} % For backslash in table

\usepackage{algorithmic}
\usepackage[linesnumbered,ruled,vlined]{algorithm2e}
\SetKwInput{KwInput}{Input}                % Set the Input
\SetKwInput{KwOutput}{Output}              % set the Output
	\SetKwProg{Fn}{Function}{}{}
	\SetKwFunction{RelaxationLMI}{Relaxation\_LMI}
	\SetKwFunction{testsatisfactionSM}{Test\_Satisfaction\_SM}
\let\oldnl\nl % Store \nl in \oldnl
\newcommand{\nonl}{\renewcommand{\nl}{\let\nl\oldnl}} % Remove line number for one line

\journal{ Journal of XXX }

\begin{document}

\begin{frontmatter}



%\title{A Quadratic Constraints approach for Critical-Time Computation in Defense Design: Application to Sharp Input Anomalies}
%\title{Critical-time Metric for Defense Design against Sharp Input Anomalies: Computation and Application on the Quadruple-Tank Case Study}
%\title{Critical-time Metric for Risk Analysis and Defense Design against Sharp Input Anomalies: Computation and Application on the Quadruple-Tank Case Study}
%\title{Critical-time Metric for Risk Analysis and Treatment against Sharp Input Anomalies: Computation and Application on the Quadruple-Tank Case Study}
%\title{Critical-time metric for risk analysis and treatment against sharp input anomalies: computation and application case study}
\title{Critical-time metric for risk analysis against sharp input anomalies: computation and application case study}


%\author{Arthur Perodou\thanks{Ampere-lab, Univ. Lyon, Ecole Centrale de Lyon, INSA Lyon, Univ. Claude Bernard Lyon 1, CNRS UMR 5005, 69130 Ecully, France} , Christophe Combastel\thanks{IMS-lab, Univ. Bordeaux, CNRS UMR 5218, 33400 Talence, France } \ and Ali Zolghadri\footnotemark[2] }

%\author[label1,label2,label2]{Arthur Perodou, Christophe Combastel and Ali Zolghadri	}
\author[label1]{Arthur Perodou\corref{cor1}}
\author[label2]{Christophe Combastel}
\author[label2]{Ali Zolghadri}
%\author[label1]{Arthur Perodou}
%\author[label2]{Christophe Combastel}
%\author[label2]{Ali Zolghadri}

\cortext[cor1]{Corresponding author. Contact: arthur.perodou@ec-lyon.fr}

\affiliation[label1]{organization={Ampere-lab, Univ. Lyon, Ecole Centrale de Lyon, INSA Lyon, Univ. Claude Bernard Lyon 1, CNRS UMR 5005},%Department and Organization
            %addressline={}, 
            city={Ecully},
            postcode={69130}, 
            %state={},
            country={France}}

\affiliation[label2]{organization={IMS-lab, Univ. Bordeaux, CNRS UMR 5218},%Department and Organization
            %addressline={}, 
            city={Talence},
            postcode={33400}, 
            %state={},
            country={France}}						

\begin{abstract}

	This paper investigates the critical-time criteria as a security metric for controlled systems subject to sharp input anomalies  (attack, fault), characterized by having high impact in a reduced amount of time (e.g. denial-of-service, attack by upper saturation).
	The critical-time is the maximal time-horizon for which a system can be considered to be safe after the occurrence of an anomaly.
	This metric is expected to be useful for risk analysis and treatment (prevention, detection, mitigation).
	In this work, the computational problem of the critical-time for uncertain linear systems and several classes of sharp input anomalies, depending on the input channel and the set of abnormal signal values, is formulated based on the quadratic constraints (QC) framework, representing sets by the intersection of QC inequalities and equalities.
	An iterative LMI-based algorithm is then proposed to provide an under-estimate of the critical-time.
	Finally, the potential of the critical-time as a metric for defense design is illustrated and discussed on the quadruple-tank case study through different relevant scenarios.

\end{abstract}



\end{frontmatter}

\section{Introduction}

	%The integration of information and communication technologies into physical systems, sometimes referred to as cyber-physical systems (CPS), has increased the vulnerability of controlled systems, as evidenced by recent cyber-attacks with significant physical impact~\cite{DPB:19}.
	The ever-increasing integration of information and communication technologies into engineering systems has led to the emergence of the cyber-physical systems (CPS) vision~\cite{Lee:15}, where cyber and physical parts are jointly co-designed.
	The Achilles heel of any CPS is widely recognized to be its fragility, i.e. its vulnerability to adverse anomalous events such as physical faults and cyber intrusions that may result in catastrophic cascading consequences and unacceptable scenarios~\cite{DPB:19}.
		%Although historically a topic of the Computer Science community, these attacks have highlighted the need to consider the interconnection between the cyber and physical parts, and it is now recognized that defense design can not be based solely on traditional IT security methods~\cite{SAJ:15}.
		Although historically a topic of the Computer Science community, these anomalous events have highlighted the need to consider the interconnection between the cyber and physical parts, and it is now recognized that defense design cannot be based solely on traditional IT security methods~\cite{SAJ:15}.
	Therefore, over the past decade, the System and Control community has developed numerous methods to address the problem of security in controlled systems~\cite{DPB:19,CST:19,WWS:16}.

	%A key point for CPS security is to develop metrics that can used efficiently for risk analysis and defense design (prevention, detection, mitigation).
	A key point for CPS security is to develop metrics that can be used efficiently for risk analysis and treatment (prevention, detection, mitigation)~\cite{TSSJ:15,CST:19}.
	%In that sense, several criteria have been developed in order to quantify the physical impact of cyber-physical attacks~\cite{FTD:14,ZhB:15,TSJ:15,TSSJ:15,MoS:16,MSR:20,BoS:21}.
	%In that sense, several quantitative criteria have been considered/developed~\cite{FTD:14,ZhB:15,TSJ:15,TSSJ:15,MoS:16,EMSZ:20,MSR:20,BoS:21}.
		%In that sense, several criteria have been developed in order to quantify the impact of cyber-physical attacks from a Control perspective~\cite{FTD:14,ZhB:15,TSJ:15,TSSJ:15,MoS:16,EMSZ:20,MSR:20,BoS:21}.
		In that sense, several criteria have been developed in order to quantify the physical impact of cyber-physical attacks.
		%~\cite{FTD:14,ZhB:15,TSJ:15,TSSJ:15,MoS:16,EMSZ:20,MSR:20,BoS:21}.
		For instance, 
			the maximum impact of stealthy attacks on the quadratic objective function of an optimal control problem 
			%while remains undetected by a detector 
			is computed and used for the design of the controller and the detector in~\cite{TSJ:15,Tei:21,ZhB:15};
			%a differential game is solved to obtain the impact of the attack in the control objective (linear quadratic problem) in~\cite{ZhB:15}, (the anomalies effect on the cost function of an optimal control problem is considered in~\cite{ZhB:15}),
			the distance of the set of reachable states by stealthy attacks to a safe/critical region is evaluated and used as a prevention tool in~\cite{TSSJ:15,EMSZ:20} or as criteria for controller and detector design in~\cite{MSR:20};
			%the maximum perturbation on the state trajectory is evaluated for stealthy attacks with minimum resources in~\cite{TSSJ:15},
	%in~\cite{MSR:20} the size of the set of reachable states by stealthy attacks and the distance to a danger set are computed and used for controller and detector re-design,
	%in~\cite{EMSZ:20} the set of reachable states by stealthy attacks on the control signal (input) and the distance to a danger set are computed and used as a prevention tool by limiting the actuators capabilities,
		the maximum number of attacks on actuators and sensors under which the state can still be estimated is proposed jointly with an associated mitigation scheme in~\cite{FTD:14};
	a metric defined by the ratio of the energy required by the attack to deviate the system from the state trajectory over the control energy needed for recovery is reported in~\cite{BoS:21}.
	%For instances, XX do that ~, YY does that ....
	%On the other hand, other metrics have been developed to quantify the ability of a system to remain safe after the occurrence of an anomaly.
	%In XX, but not mean to compute it. In Escudero: distance to danger zone for aging attack, a particular stealthy attack (+ limited to ellipsoidal set ?).
	
	Recently, a new security metric named critical-time was proposed in~\cite{PCZ:21}.
	%The critical-time is the maximal time-horizon for which a system is considered to be safe after the occurrence of an anomaly (attack, fault), that is the system is not in a critical state and is still able to recover to a nominal mode.
	The critical-time is the maximal time-horizon for which a system is considered to be safe after the occurrence of an anomaly (attack, fault), that is the system is not in a critical state allowing to recover to a nominal mode or to  reconfigure for a safe stop.
	The interest of the critical-time is manifold.
	First, this time-based metric can be usefully included in each stage of the risk management framework.
	In the risk analysis stage, it identifies system vulnerabilities and anomalies with the greatest impact, associated with a low critical-time.
	In the risk treatment stage, this metric may be used to prevent anomalies, associated then with high or even infinite critical-time~\cite{EMSZ:20}, to provide a design objective for the synthesis of anomaly detectors~\cite{HKKS:10,GWS:19} or for the development of a mitigation mechanism, such as reconfigurable control~\cite{ZhJ:08}, dissipation block activation~\cite{PCZ:21b} or software rejuvenation~\cite{RGKS:20}.
	Moreover, maximizing the critical-time will leave extra-time for defense mechanisms, including human operators, to detect and mitigate anomalies.
	This may be achieved by tuning remaining degrees of freedom, as in a rational security allocation approach~\cite{CST:19}.
	Furthermore, an originality is to consider the more general concept of anomaly, instead of only attacks, as the impact of certain attacks and faults can be modeled in a similar way (see for instance \cite{YeL:19}).
	This especially allows this metric to be integrated in a secure and safe approach where both attacks and faults are considered.
	In particular, the critical-time generalizes the detection and reconfiguration delays classically used in the fault-management literature~\cite{HKKS:10,ZhJ:08}.

	%In this paper, the critical-time for sharp anomalies input, that is having high impact in minimum time (e.g. deny-of-service, attack by upper saturation), is investigated.
	
	In this paper, the critical-time metric for controlled systems subject to sharp input anomalies, characterized by having high impact in a reduced amount of time (e.g. denial-of-service, attack by upper saturation), is investigated.
	%Unlike the more sophisticated stealthy attacks, sharp anomalies
		The associated computational problem for an uncertain linear discrete-time model is addressed by taking advantage from the quadratic constraints (QC) framework, representing sets by the intersection of QC inequalities and equalities.
		This problem is reformulated in terms of LMI (Linear Matrix Inequality)~\cite{BEFB:94} and an iterative LMI-based algorithm is proposed to provide an under-estimate of the critical-time.
		Unlike~\cite{PCZ:21}, limited to the worst-case scenario where the attacker has full control of the inputs, the class of sharp input anomalies is significantly extended.
		By explicitly modeling the controller, sharp anomalies that only impact some inputs, and not all of them, can be considered.
		In addition, the inclusion of quadratic constraint equalities (QCE) allows for considering specific sharp anomalies, such as denial-of-service (DoS) or attack by upper saturation.
		%, and enables to generalize the computation to larger classes of convex sets.
			Moreover, this paper relaxes the assumption on the initial state which is assumed now to be uncertain, leading to a more practical and robust result.
		%Finally, the potential of the critical-time metric for defense design is thoroughly illustrated on the quadruple-tank case study~\cite{Joh:00} through three scenarios: 1) an anomaly-specific scenario for the DoS and attack-by-upper saturation anomalies, 2) a worst-case scenario where the applied input can take any values admissible by the actuators, 3) a channel-dependent scenario where the anomalies only impact a single input. 
		Finally, the potential of the critical-time metric for defense design is thoroughly illustrated on the quadruple-tank case study~\cite{Joh:00} through different scenarios.
		%: 1) an anomaly-specific scenario for the DoS and attack-by-upper saturation anomalies, 2) a worst-case scenario where the applied input can take any values admissible by the actuators, 3) a channel-dependent scenario where the anomalies only impact a single input. 
		%In particular, it is shown how the critical-time can be maximized by suitably selecting the remaining degrees of freedom, 
		%%e.g. here valve parameters,
		%leaving then extra-time for defense mechanisms.
		In particular, it is shown how the critical time can be maximized by properly selecting the remaining degrees of freedom, leaving extra-time for defense mechanisms.

		The paper is organized as follows.
		In Section~\ref{sec_convex_set_by_QCs}, the QC framework for convex set representations is introduced.		
		In Section~\ref{sec_Pb_statement}, the critical-time computation problem of a controlled system, with an uncertain discrete-time linear model, subject to sharp input anomalies is formulated.
		Section~\ref{sec_Main_result} demonstrates how the computation can be performed based on LMI optimization.
		Finally, the approach is applied on the quadruple-tank case study in Section~\ref{sec_illustration}, while {Section~\ref{sec_conclusion} provides some concluding remarks}.

%
%\subsection{Related work}
%
	%Based on \cite{PCZ:21}
	%
	%$\rightarrow$ Link/difference with usual reachability analysis/safety verification techniques ?

%\subsection*{Notations}
\paragraph*{Notations}
		
				%$\rightarrow$ Define $\mathbb{R}^+$
		
				Lower (upper) case letters are used for vectors (matrices). 
	%$\mathbb{N}$ denotes the set of natural numbers, $\mathbb{Z}$ the set of integers, $\llbracket k_1, k_2 \rrbracket$ the integers between, and including, $k_1$ and $k_2$, $\mathbb{R}^+$ the set of non-negative real numbers and $\mathbb{R}^{n\times m}$ the set of real-valued matrices of size $n\times m$.
	$\mathbb{R}^+$ denotes the set of non-negative real numbers and $\mathbb{R}^{n\times m}$ the set of real-valued matrices of size $n\times m$.
	$I_n$ and $0_{n \times m}$ are respectively the identity matrix of $\mathbb{R}^{n\times n}$ and the zero matrix of $\mathbb{R}^{n\times m}$. The subscripts are omitted when obvious from the context.
	$X^T$ stands for transpose of $X$ while $M > (\text{resp. } \geq)~0$ denotes positive (semi-) definiteness. 	
	For the sake of brevity, $\llbracket k_1, k_2 \rrbracket$ refers to the integers from~$k_1$ to~$k_2$ $z_{\llbracket k_1, k_2 \rrbracket} := \begin{bmatrix} z_{k_1}^T & \hdots & z_{k_2}^T \end{bmatrix}^T$.
	%$\mathrm{trace}(M)$ is the sum of the diagonal elements of $M$ and $\mathrm{diag}(M)$ is their concatenation in a column vector.
	%The sign $\otimes$ represents the Kronecker product.
	%Bold characters denote either explicit decision variables in a design problem or optimization variables in an optimization problem.
	Bold characters denote decision variables in a design problem or optimization variables in an optimization problem.

\section{Convex Set Representation by Quadratic Constraints Intersection} \label{sec_convex_set_by_QCs}
%\section{Preliminaries}
%\label{sec_convex_set_by_QCs}

		This section introduces and motivates the choice of Quadratic Constraints (QC) as a computational framework. 
		First, QC can be used to represent a large class of convex sets.
		In addition, this framework allows for considering the intersection of these sets, without explicitly computing it.
		Finally, its link with LMI optimization provides efficient numerical methods.

		%In this paper, sets will be described using the quadratic constraint (QC) framework. 
		%%This framework is general enough to exactly represent the majority of convex sets, such as ellipsoids, polytopes and zonotopes, as displayed in this section.
		%This framework is general enough to represent the majority of convex sets, such as ellipsoids, polytopes and zonotopes, as displayed in this section.
		%Moreover, it allows for considering the intersection of these sets, without explicitly computing the intersection.
		%%%Finally, its linked with LMI optimization leads to a numerical method that appears efficient from the perspective of the (offline) synthesis/design of the detector~$(\mathcal{D})$ or the controller~$(\mathcal{K})$.
		%%Finally, its link with LMI optimization leads to a numerical method that appears efficient from the perspective of the offline design of a defense strategy.

		%\subsection{QC definition and S-procedure}
	
	%First, the following definition of a QC is introduced.
	%\begin{defin}[Vector-based QC] \label{def_QC} \ \\
	%\begin{defin}[Quadratic constraint \cite{Jon:01b}] \label{def_QC} \ \\
	\begin{defin} \label{def_QC} \ \\
			Consider a quadratic function~$\sigma_S(\cdot)$ defined by:
			\begin{equation*}
					%\forall z \in \mathbb{R}^{n_z}, \qquad \sigma_S(z) := \begin{bmatrix} z \\ 1 \end{bmatrix}^T S \begin{bmatrix} z \\ 1 \end{bmatrix}
					\forall z \in \mathbb{R}^{n_z}, \qquad \sigma_S(z) := \begin{bmatrix} z \\ 1 \end{bmatrix}^T S \begin{bmatrix} z \\ 1 \end{bmatrix}
			\end{equation*}
			where $S = S^T \in \mathbb{R}^{(n_z+1)\times (n_z+1)}$.
			A vector $z \in \mathbb{R}^{n_z}$ is said to satisfy a QC associated with~$\sigma_S$ if $\sigma_S(z) \geq 0$	holds.
			In addition, $z \in \mathbb{R}^{n_z}$ is said to satisfy a QC equality (QCE) associated with~$\sigma_S$ if $\sigma_S(z) = 0$	holds.
	\end{defin}

		%It can be inferred from Definition~\ref{def_QC} that the more common classes of convex sets classically used for reachability analysis~\cite{AFG:21} can be exactly expressed in the QC framework.
		It can be inferred from Definition~\ref{def_QC} that the most common classes of convex sets classically used~\cite{AFG:21} can be exactly expressed in the QC framework.
		%
			%
			%%The QC framework enables to represent most of (convex) set used in the literature.	
			%Most of convex sets traditionally used~\cite{AFG:21} can be exactly represented in the QC framework.		
			For example, an ellipsoid and a degenerated ellipsoid, such as an halfspace, is inherently represented as a single QC~\cite{Jon:01b}.
			%Consider the following set:
			The set $\mathcal{E}$ defined by
			\begin{equation*}
				\mathcal{E} := \left\{ z \in \mathbb{R}^{n_z} \left| \sigma_{P_{\mathcal{E}}}(z) \geq 0 \right. \right\}, \quad P_{{\mathcal{E}}} := \begin{bmatrix} Q & s \\ s^T & r \end{bmatrix}
			\end{equation*}
			leads to an ellipsoid when~$Q < 0$, a degenerated ellipsoid when~$Q \leq 0$ and a halfspace when~$Q = 0$.
			The intersection of QCs offers then the possibility to represent the intersection of ellipsoids and degenerated ellipsoids.
			%In particular, the intersection of halfspaces enables to exactly represent polytopes, and naturally their intersection.
			In particular, a polytope in~$H$-representation
			$\mathcal{H} := \left\{ z \in \mathbb{R}^{n_z} \left| \ Hz \leq h  \right. \right\}$, $h \in \mathbb{R}^{n_h}$, 
			%and their intersection
			is represented by the intersection of halfspaces:
			%%
			%%Furthermore, the intersection of QCs also directly leads to an exact representation of polytopic sets, and naturally their intersection.
			%Indeed, a polytope in~$H$-representation
			%$\mathcal{H} := \left\{ z \in \mathbb{R}^{n_z} \left| \ Hz \leq h  \right. \right\}$
			%%\begin{equation*}
					%%\mathcal{H} := \left\{ z \in \mathbb{R}^{n_z} \left| \ Hz \leq h  \right. \right\}
			%%\end{equation*}
			%%where~$H \in \mathbb{R}^{n_h \times n_x}$ and~$h \in \mathbb{R}^{n_h}$, 
			%is the intersection of~$n_h$ halfspaces:
			\begin{equation*}
					\mathcal{H} = \bigcap_{i=1}^{n_h}{\left\{ z \in \mathbb{R}^{n_z} \left| \ \sigma_{P_{\mathcal{H}_i}}(z) \geq 0 \right. \right\}}, \quad
					P_{\mathcal{H}_i} := \begin{bmatrix} 0 & H_i^T \\ H_i & 2h_i \end{bmatrix}
			\end{equation*}	
				%where	$H_i$ is the~$i$th line of the matrix~$H \in \mathbb{R}^{n_h \times n_z}$ and~$h_i$ the~$i$th element of the vector~$h \in \mathbb{R}^{n_h}$.
				%where	$H_i$ is the~$i$th line of~$H \in \mathbb{R}^{n_h \times n_z}$ and~$h_i$ the~$i$th element of~$h \in \mathbb{R}^{n_h}$.
				where	$H_i$ the~$i$th line of~$H$ and~$h_i$ the~$i$th element of~$h$.
				Moreover, a QCE especially enables to represent convex sets defined by a linear equality, such as a hyperplane
				$\mathcal{HP} := \left\{ z \in \mathbb{R}^{n_z} \left| \ a^Tz = b  \right. \right\}$
			%\begin{equation*}
					%\mathcal{HP} := \left\{ z \in \mathbb{R}^{n_z} \left| \ a^Tz = b  \right. \right\}
			%\end{equation*}
			where $a \in \mathbb{R}^{n_z}$ and ${b \in \mathbb{R}}$. 
			Planes, that are the intersection of hyperplanes, can then be represented by the intersection of QCEs.
			%The intersection of QCEs leads then to the representation of hyperplanes intersection, i.e. planes.
			%
			Finally, considering the intersection of both QCs and QCEs significantly extends the scope. For instance, a zonotope
			$\mathcal{Z} = \left\{ z \in \mathbb{R}^{n_z} \left| z = c + G\lambda, \ \lambda \in [-1,1]^{n_{\lambda}} \right. \right\}$
					%\begin{equation*}
						%\mathcal{Z} = \left\{ z \in \mathbb{R}^{n_z} \left| z = c + G\lambda, \ \lambda \in [-1,1]^{n_{\lambda}} \right. \right\}
					%\end{equation*}			
				can be expressed as the intersection of $n_z$~ hyperplanes and $n_{\lambda}$~degenerated ellipsoids, i.e. QCEs and QCs.
				%, if the augmented vector~$\overline{z}^T := \begin{bmatrix} z^T & \lambda^T \end{bmatrix}$ in the augmented space~$\mathbb{R}^{n_z + n_{\lambda}}$ is considered.
				Zonotope bundles, that are sets defined as the intersection of zonotopes, are also exactly represented, as well as constrained zonotopes
				${\mathcal{CZ} = \left\{ z \in \mathbb{R}^{n_z} \left| z = c + G\lambda, \ D \lambda + d = 0, \ \lambda \in [-1,1]^{n_{\lambda}}  \right. \right\}}$
				%\begin{equation*}
						%\mathcal{CZ} = \left\{ z \in \mathbb{R}^{n_z} \left| z = c + G\lambda, \ D \lambda + d = 0, \ \lambda \in [-1,1]^{n_{\lambda}}  \right. \right\}
				%\end{equation*}
				which result from the intersection of a zonotope with hyperplanes.
				%Up to the numerical reliability of related optimization solvers, QCs provide a useful design framework to exactly define and operate over a wide class of convex sets. 
				Up to the numerical reliability of optimization solvers, QCs provide a useful design framework to exactly define and operate over a wide class of convex sets. 
				
				%\begin{remark}
						%%Remark sur zonotope: particular polytope (mais compactness ...) + placement vis-à-vis de Fazlyab
						%A zonotope is a particular polytope and can be also represented in a H-representation, but this comes at the price of the compactness of the representation.
						%It should be also noticed that an over-approximation of a zonotope in the QC framework has been recently reported in~\cite{FMP:20}.
				%\end{remark}

	%An important result of the QC framework is the S-procedure lemma that provides a sufficient condition, in the form of an LMI feasibility problem, to test the satisfaction of a QC by a set of vectors defined by the intersection of QCs~\cite[Chap.~2]{BEFB:94}.
	An important result of the QC framework is the S-procedure lemma that provides a sufficient condition, in the form of a feasibility optimization problem under LMI constraints, to test the satisfaction of a QC by a set of vectors defined by the intersection of QCs~\cite[Chap.~2]{BEFB:94}.
	%An important result of the QC framework is the S-procedure lemma that provides a sufficient condition to test the satisfaction of a QC by a set of vectors defined by the intersection of QCs~\cite[Chap.~2]{BEFB:94}.
		The next lemma provides an adapted version when a QC is checked over multiple QCs and QCEs~(see \cite[Theorem~2.3.3]{Sco:97}).
	
	%QC framework ==> S-procedure ==> LMI feasibility problem
	%
	%adapted version here, with a straightforward proof.
	
	\begin{lem}[S-procedure for intersection of QCs and QCEs] \label{lem_Sprocedure} \ \\
	%\begin{lem}[S-procedure for QCs and QCEs \cite[Theorem 2.3.3]{Sco:97} ] \label{lem_Sprocedure} \ \\
			Consider $(N+1)$ QCs associated with~$\sigma_S$, $\sigma_{P_1}$, $\ldots$, $\sigma_{P_N}$, and~$M$ QCEs associated with~$\sigma_{Q_1}$, $\ldots$, $\sigma_{Q_M}$.
			\\
			Then $(ii) \Rightarrow (i)$.
		\begin{itemize}
			\item[$(i)$] The QC $\sigma_S(z) \geq 0$ holds for all $z \in \mathbb{R}^{n_z}$ such that
				\begin{equation*}
						\begin{aligned}
								\forall p \in \llbracket 1, N \rrbracket & \quad & \sigma_{P_p}(z) \geq 0 \\
								\forall q \in \llbracket 1, M \rrbracket & \quad & \sigma_{Q_q}(z) = 0 
						\end{aligned}
				\end{equation*}
			\item[$(ii)$] $\exists \left\{\boldsymbol{\tau_{p}} \in \mathbb{R}^+\right\}_{p=1}^{N}$, $\exists \left\{ \boldsymbol{\mu_{q}} \in \mathbb{R} \right\}_{q=1}^{M}$, such that
			\begin{equation}
					S - \sum_{p=1}^{N}{\boldsymbol{\tau_p}P_p} - \sum_{q=1}^{M}{\boldsymbol{\mu_q}Q_q} \geq 0
					\label{eq_Sproc_lemma}
			\end{equation}
			%holds.
		\end{itemize}
	\end{lem}		
	
	\begin{remark}
			The scalars $\boldsymbol{\tau_p}$ and $\boldsymbol{\mu_q}$, classically called multipliers, are the optimization variables in the constraint (\ref{eq_Sproc_lemma}), leading then to a feasibility problem over LMI constraints.
	\end{remark}
			% proof
			% Multiplier
			% $\rightarrow$ allows to check quadratic constraint over other QCs. Need to reformulate the problem like that in the rest of the paper.

%\input{Files/2_Pb_statement.tex}

%\section{Problem Statement} \label{sec_Pb_statement}
\section{Problem Setting} \label{sec_Pb_statement}
%\section{Problem Formulation} \label{sec_Pb_statement}

		%\begin{figure}[!htb]
      %\centering
      %\includegraphics[width=0.75\columnwidth]{Dissipation_block}
      %\caption{Interconnection of~$G$ with the dissipation block~$M$}
      %\label{fig_dissipation_block}
   %\end{figure}		

		%\begin{itemize}
			%%\item 2 types of problems: critical-time computation problems with fixed or free initial state
			%\item Mettre une remarque quelque part que le calcul du critical-time a pour vocation de se faire offline ? Donc on peut potentiellement faire mieux en terme de calcul mais nous on cherche juste a pouvoir mixer les representations et leur intersection ?
		%\end{itemize}

		%\subsection{Nominal System Model}
		%\subsection{Controlled system model under normal conditions}
		\subsection{Nominal system description}
		
			Consider an uncertain linear discrete-time system~$(\Sigma)$:
		\begin{equation*}
		(\Sigma)
			\left\{
				\begin{aligned}
				x_{k+1} &= A_{\Sigma} x_k + B_{\Sigma} u_k^a + W_{\Sigma}w_k,
			\quad
			u_k^a \in \mathcal{U} \\
				y_k &= C_{\Sigma} x_k + V_{\Sigma} v_k 
				\end{aligned}				
			\right.
			%\label{eq_system_model}
		\end{equation*}
		%\begin{equation}
		%(\Sigma)
			%\left\{
				%\begin{aligned}
				%x_{k+1} &= A_{\Sigma} x_k + B_{\Sigma} u_k^a + W_{\Sigma}w_k \\
				%y_k &= C_{\Sigma} x_k + v_k 
				%\end{aligned}				
			%\right.,
			%\quad
			%u_k^a \in \mathcal{U}
			%\label{eq_system_model}
		%\end{equation}							
		where $x_k \in \mathbb{R}^n$ is the system state evaluated at time step~$k$, $u_k^a \in \mathbb{R}^m$ the applied input and $y_k$ the measured output. 
		%Due to physical limitations of the inputs, such as saturation or slew rate, the applied input~$u_{a_k}$ is constrained to belong to the input-limitation set~$\mathcal{U} \subseteq \mathbb{R}^{m}$.
		The applied input~$u_k^a$ is constrained to belong to the input-limitation set~$\mathcal{U} \subseteq \mathbb{R}^{m}$.
		%The vectors~$w_k \in\mathbb{R}^{w}$ and~$v_k \in \mathbb{R}^{v}$ denote additive uncertainties, including for instance modeling approximation or measurement noises. 
		The vectors~$w_k \in\mathbb{R}^{n_w}$ and~$v_k \in \mathbb{R}^{n_v}$ represent additive state disturbances and measurement noises. 

		The system~$(\Sigma)$ is interconnected, through a communication network, to a controller~$(\mathcal{K})$ modeled as:
		\begin{equation*}
		(\mathcal{K})
			\left\{
				\begin{aligned}
				\xi_{k+1} &= A_K \xi_k + B_K y_k \\
				u_k &= C_K \xi_k + D_K y_k 
				\end{aligned}
			\right.
			%\label{eq_system_model}
		\end{equation*}			
		%where $u_k \in \mathbb{R}^m$ is the healthy control input and $\xi_k \in \mathbb{R}^l$ is the controller state.
		where $\xi_k \in \mathbb{R}^l$ is the controller state and $u_k \in \mathbb{R}^m$ is the controller output.
		Under normal conditions, $u_k^a = u_k$.
		
		\subsection{Sharp input anomalies}
		
	%In this paper, sharp anomalies, i.e. having severe impact in a reduced amount of time, on inputs are considered.
	In this paper, sharp anomalies, producing severe impact in a reduced amount of time, on inputs are considered.
	This especially includes anomalies such as sudden actuator failures, DoS and attack by saturation.
	From a security point of view, sharp attacks result typically from adversaries with little knowledge of the overall system, and may be accomplished from different manners (fragility of a communication network protocol, communication channel jamming,  direct attack on the controller, or even resulting from false-data measurements)~\cite{SRP:19}.
	%From a security point of view, sharp attacks result typically from adversaries with little knowledge of the overall system, that  may take advantage of differet fragilities of communication
	%from different manners (fragility of a communication network protocol, communication channel jamming,  direct attack on the controller, or even resulting from false-data measurements ...)~\cite{SRP:19}.
		%
		%and may take advantage of differet fragilities of communication
%(network protocol, communication channel jamming ...) or result from
%direct attack on the controller, or  fals edata measurements. See [25].
		%
		Unlike the more sophisticated class of stealthy attacks, sharp anomalies may be detected by a passive detection scheme, for instance based on an observer or using a $\chi^2$-detector~\cite{DPB:19,CST:19}.		
		
			%The system~$(\Sigma)$ is assumed to be subject to input anomalies, modeled as
			The anomalies impact on the applied input is modeled as
			\begin{equation*}
					u_k^a = \Gamma_u u_k + \Gamma_a a_k
			\end{equation*}
	where
	%$a_k \in \mathcal{A}$ is an unknown signal representing the impact of the anomalies, with~$\mathcal{A} \subseteq \mathbb{R}^{m_a}$ the set of possible values for~$a_k$, 
	$a_k \in \mathcal{A}$ is an unknown signal and~$\mathcal{A} \subseteq \mathbb{R}^{m_a}$ the set of possible values for~$a_k$, 
	%while the matrices~$\Gamma_u \in \mathbb{R}^{m \times m}$ and~$\Gamma_a \in \mathbb{R}^{m \times m_a}$ denote the input channels on which the anomalies appear.
	while~$\Gamma_u \in \mathbb{R}^{m \times m}$ and~$\Gamma_a \in \mathbb{R}^{m \times m_a}$ denote the input channels on which the anomalies appear, and are typically binary-valued $\{0,1\}$-matrices.
	%The anomalies effect is then modeled by the triplet~$(\Gamma_u, \Gamma_a, \mathcal{A})$.
	%An important number of class of sharp input anomalies are modeled by the triplet~$(\Gamma_u, \Gamma_a, \mathcal{A})$.
	A large class of sharp input anomalies can be modeled by the triplet~$(\Gamma_u, \Gamma_a, \mathcal{A})$.
	For instances, 
		$(\Gamma_u, \Gamma_a, \mathcal{A}) = (0,I,\mathcal{U})$ models the worst-case scenario where an adversary may apply any admissible input,
		%$(\Gamma_u, \Gamma_a, \mathcal{A}) = (0,I,0)$ often result from a DoS,
		$(\Gamma_u, \Gamma_a, \mathcal{A}) = (0,I,0)$ may be associated with a DoS,
		%$(\Gamma_u, \Gamma_a, \mathcal{A}) = (I,I,\{a = \alpha \cdot u_{max}\})$, where $u_{max}$ the maximal admissible applied input and $\alpha \in \mathbb{R}$ a potentially large number, models additive actuator faults or data-injection attacks with saturation objective, 
		%$(\Gamma_u, \Gamma_a) = (I,I)$ models additive actuator faults or data-injection attacks, that are with saturation objective when $\mathcal{A} = \{a = \alpha \cdot u_{max}\}$ where $u_{max}$ the maximal admissible applied input and $\alpha \in \mathbb{R}$ a potentially large number, 
		$(\Gamma_u, \Gamma_a) = (I,I)$ models additive input faults or data-injection attacks, that are with saturation objective when $\mathcal{A} = \{a = \alpha \cdot u^a_{max}\}$ where $u^a_{max}$ the maximum admissible  input and $\alpha \in \mathbb{R}$ a large number, 
		while 
		%$\Gamma_u = \begin{bmatrix} I & 0 & 0 \\ 0 & 0_{1\times 1} & 0 \\ 0 & 0 & I \end{bmatrix}$ and $\Gamma_a = \begin{bmatrix} 0 & 0 & 0 \\ 0 & 1 & 0 \\ 0 & 0 & 0 \end{bmatrix}$ 
		\begin{equation*}
				\Gamma_u = \begin{bmatrix} I & 0 & 0 \\ 0 & 0_{1\times 1} & 0 \\ 0 & 0 & I \end{bmatrix}
				\quad
				\text{and}
				\quad
				\Gamma_a = \begin{bmatrix} 0 & 0 & 0 \\ 0 & 1 & 0 \\ 0 & 0 & 0 \end{bmatrix}
		\end{equation*}
		model anomalies that appear on the $i ^{th}$ input channel.	
	%$(\Gamma_u,\Gamma_a,\mathcal{A}) = (0,I,\mathcal{U})$
	%+ data injection attack 
	Normal operating conditions are obtained when~$(\Gamma_u,\Gamma_a) = (I,0)$.

	\subsection{Controlled system subject to sharp input anomalies}

		\begin{figure}[!htb]
      \centering
      \includegraphics[width=0.75\columnwidth]{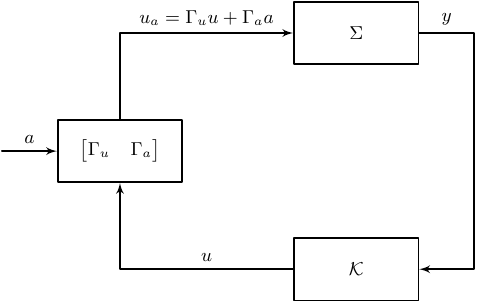}
      \caption{Controlled system subject to sharp input anomalies}
      \label{fig_actuator_anomalies}
   \end{figure}

		%Defining $\overline{x}_k := \begin{bmatrix} x_k^T & \xi_k^T \end{bmatrix}^T$ and $\overline{w}_k := \begin{bmatrix} w_k^T & v_k^T \end{bmatrix}^T$, the dynamics of the interconnected system~$(\Sigma-\mathcal{K})$ is given by:
		%Defining $\overline{x}_k := \begin{bmatrix} x_k^T & \xi_k^T \end{bmatrix}^T$ and $\overline{w}_k := \begin{bmatrix} w_k^T & v_k^T \end{bmatrix}^T$, 
		The dynamics of the interconnected system~$(\Sigma-\mathcal{K})$ subject to input anomalies (Fig.~\ref{fig_actuator_anomalies}) is then given by:
		\begin{equation}
				(\Sigma-\mathcal{K})
				\left\{
					\begin{aligned}
					\overline{x}_{k+1} &= A \overline{x}_k + W \overline{w}_k + B_a a_k 
					\\
					%u^a_k &= \Gamma_uC \overline{x}_k + \Gamma_u V \overline{w}_k + \Gamma_a a_k, 
					u^a_k &= C_u \overline{x}_k + V_u \overline{w}_k + \Gamma_a a_k, 
					\quad
					u^a_k \in \mathcal{U}
					\end{aligned}
				\right.
				\label{eq_SKA}
		\end{equation}		
				%with $\overline{x}_k = \begin{bmatrix} x_k^T & \xi_k^T \end{bmatrix}^T$, 
				%$A  = \begin{bmatrix} A_{\Sigma} +  B_{\Sigma} \Gamma D_K C_{\Sigma} & A_{\Sigma} +  B_{\Sigma} \Gamma C_K  \\ B_K C_{\Sigma} & A_K \end{bmatrix}$
				where $\overline{x}_k := \begin{bmatrix} x_k^T & \xi_k^T \end{bmatrix}^T$, $\overline{w}_k := \begin{bmatrix} w_k^T & v_k^T \end{bmatrix}^T$ and
		\begin{equation*}
				A  = \begin{bmatrix} A_{\Sigma} +  B_{\Sigma} \Gamma_u D_K C_{\Sigma} & B_{\Sigma} \Gamma_u C_K  \\ B_K C_{\Sigma} & A_K \end{bmatrix},
				\
				B_a = \begin{bmatrix} B_{\Sigma} \Gamma_a \\ 0 \end{bmatrix}
				%\
				%W = \begin{bmatrix} W_{\Sigma} & B_{\Sigma}\Gamma_uD_K V_{\Sigma} \\ 0 & B_K V_{\Sigma} \end{bmatrix}
		\end{equation*}
		\begin{equation*}
				W = \begin{bmatrix} W_{\Sigma} & B_{\Sigma}\Gamma_uD_K V_{\Sigma} \\ 0 & B_K V_{\Sigma} \end{bmatrix}, \quad		
			\begin{aligned}
					C_u &= \Gamma_u\begin{bmatrix} D_K C_{\Sigma} & C_K \end{bmatrix} \\
					V_u &= \Gamma_u\begin{bmatrix} 0 & D_K V_{\Sigma}  \end{bmatrix}
			\end{aligned}
		\end{equation*}
	%This interconnected system~$(\Sigma - \mathcal{K})$ is assumed to be provided with a safety set~$\overline{\mathcal{S}} \subseteq \mathbb{R}^{n+l}$, i.e. a set of states~$\overline{x}_k$ for which~$(\Sigma-\mathcal{K})$ is considered to be safe at each time step~$k$. 
	This interconnected system~$(\Sigma - \mathcal{K})$ is provided with a safety set~$\overline{\mathcal{S}} \subseteq \mathbb{R}^{n+l}$, i.e. a set of states~$\overline{x}_k$ for which~$(\Sigma-\mathcal{K})$ is considered to be safe at each time step~$k$. 
		%%The safety set~$\mathcal{S}$ results from physical constraints or the validity domain of a linearized model for instance.
	%The safety set~$\overline{\mathcal{S}}$ is usually determined by physical constraints on states.	
	For the ease of computation, and without loss of generality, it is assumed that the anomalies start from~$k_a=0$.
	The initial state~$\overline{x}_0$ is assumed to belong to a known and safe set~$\overline{\mathcal{X}}_0 \subseteq \overline{\mathcal{S}}$, and the additive uncertainties~$\overline{w}_k$ to
	%a bounded set~$\overline{\mathcal{W}} \subseteq \mathbb{R}^{n_{\overline{w}}}$, where $n_{\overline{w}} = n_w + n_v$.
	a bounded set~$\overline{\mathcal{W}} \subseteq \mathbb{R}^{n_w + n_v}$.
	%with $y = v + w$.

	The defense objective is to design off-line a strategy such that safe operation is preserved after the occurrence of an anomaly.
	%In this perspective, an (under)-estimation of the critical-time, that is the maximum time window allowed before the system leaves the safety set, appears as a crucial information to be provided to the defender.
	%In this perspective, an (under)-estimation of the critical-time, that is the maximum time window allowed before the system leaves the safety set, appears as a crucial information to be provided to the defense designer.
	In this perspective, an (under)-estimation of the critical-time, that is the maximum time window allowed before the system leaves the safety set, appears as a crucial information to be provided to the defense designer.

		%The defense objective is to preserve safety after the occurrence of sharp actuator anomalies.
		%%, by synthesizing suitable anomaly detectors and mitigation strategies for instance.
		%In this perspective, the estimation of the maximum time allowed to the defender before the system leaves the safety set appears as a crucial information, provided by the critical-time computation.
%
				%
		%%The critical time~$k_c$ is the maximal time-horizon such that, after the occurrence of an anomaly, the system remains safe over the time-window~$\llbracket 0, k_c\rrbracket$.
		%The critical-time~$k_c$ is the maximal time-horizon such that, after the occurrence of an anomaly, the system remains safe over the time-window~$\llbracket k_a, k_a + k_c\rrbracket$.
		%The main benefit from availability of the critical-time~$k_c$ is the estimation of the time 
			%that can be allowed for defenses, such as a human operator or an automatic mechanism.
			%In the latter case, this especially provides an upper-bound for the detection and reconfiguration delay.

		%\subsection{Critical-time metric}
		%\subsection{Computational problem based on the QC framework}
		\subsection{Problem formulation}

	The problem considered in this paper is the computation of the critical-time for the interconnected system~(\ref{eq_SKA}) subject to sharp input anomalies defined by $(\Gamma_u, \Gamma_a, \mathcal{A})$, based on the QC framework.	
	To this end,
	%Moreover, 
	the safety set~$\overline{\mathcal{S}}$ is assumed to be exactly represented or inner-approximated by the intersection of~QCs (Definition~\ref{def_QC}):
	\begin{equation}
			\bigcap_{s=1}^{n_{\overline{\mathcal{S}}}}{\overline{\mathcal{S}}_s} \subseteq \overline{\mathcal{S}},
			%\bigcap_{i=1}^{n_{\mathcal{S}}}{\left\{ x \in \mathbb{R}^{n} \left| \sigma_{S_i}(x) \geq 0 \right. \right\}} \subseteq \mathcal{S}
			\qquad
			\overline{\mathcal{S}}_s := \left\{ \overline{x} \in \mathbb{R}^{n+l} \left| \sigma_{\overline{S}_s}(\overline{x}) \geq 0 \right. \right\}
			\label{eq_def_inter_sets_safety}
	\end{equation}	
	%while the initial-state set~$\overline{\mathcal{X}}_0$, the input-limitation set~$\mathcal{U}$, the deviation set~$\overline{\mathcal{W}}$ and the anomalies impact set~$\mathcal{A}$ are assumed to be exactly represented or over-approximated by the intersection of QCs and QCEs:
	and the initial set~$\overline{\mathcal{X}}_0$, the input-limitation set~$\mathcal{U}$, the deviation set~$\overline{\mathcal{W}}$ and the anomaly-impact set~$\mathcal{A}$ are exactly represented or over-approximated by the intersection of QCs and QCEs:
	\begin{equation}
			\overline{\mathcal{X}}_0 \subseteq \left\{ \cap_{g=1}^{n_{\overline{\mathcal{X}}_0}}{\ \overline{\mathcal{X}}_g} \right\} \bigcap \left\{ \cap_{h=1}^{m_{\overline{\mathcal{X}}_0}}{\ \overline{\mathcal{X}}_h^e} \right\} 
			\label{eq_def_inter_sets_initial_set}
	\end{equation}	
	\begin{equation}
			\mathcal{U} \subseteq \left\{ \cap_{i=1}^{n_{\mathcal{U}}}{\ \mathcal{U}_i} \right\} \bigcap \left\{ \cap_{j=1}^{m_{\mathcal{U}}}{\ \mathcal{U}_j^e} \right\} 
			\label{eq_def_inter_sets_input}
	\end{equation}
	\begin{equation}
			\overline{\mathcal{W}} \subseteq \left\{ \cap_{p=1}^{n_{\overline{\mathcal{W}}}}{\ \overline{\mathcal{W}}_p} \right\} \bigcap \left\{ \cap_{q=1}^{m_{\overline{\mathcal{W}}}}{\ \overline{\mathcal{W}}_q^e} \right\}
			\label{eq_def_inter_sets_dev}
	\end{equation}
	\begin{equation}
			\mathcal{A} \subseteq \left\{ \cap_{r=1}^{n_{\mathcal{A}}}{\ \mathcal{A}_r} \right\} \bigcap \left\{ \cap_{t=1}^{m_{\mathcal{A}}}{\ \mathcal{A}_t^e} \right\}
			\label{eq_def_inter_sets_A}
	\end{equation}
%%% TEST
	%\begin{equation}
			%\overline{\mathcal{X}}_0 \subseteq \left\{ \cap_{b=1}^{n_{\overline{\mathcal{X}}_0}}{\ \mathcal{X}_b} \right\} \bigcap \left\{ \cap_{c=1}^{m_{\mathcal{X}_0}}{\ \mathcal{X}_c^e} \right\} 
			%%\label{eq_def_inter_sets_input}
	%\end{equation}
	%\begin{equation}
			%\mathcal{W} \subseteq \left\{ \cap_{g=1}^{n_{\mathcal{W}}}{\ \mathcal{W}_g} \right\} \bigcap \left\{ \cap_{h=1}^{m_{\mathcal{W}}}{\ \mathcal{W}_h^e} \right\}
			%%\label{eq_def_inter_sets_dev}
	%\end{equation}
	%\begin{equation}
			%\mathcal{A} \subseteq \left\{ \cap_{i=1}^{n_{\mathcal{A}}}{\ \mathcal{A}_i} \right\} \bigcap \left\{ \cap_{j=1}^{m_{\mathcal{A}}}{\ \mathcal{A}_j^e} \right\}
			%%\label{eq_def_inter_sets_A}
	%\end{equation}			
	%\begin{equation}
			%\mathcal{U} \subseteq \left\{ \cap_{p=1}^{n_{\mathcal{U}}}{\ \mathcal{U}_p} \right\} \bigcap \left\{ \cap_{q=1}^{m_{\mathcal{U}}}{\ \mathcal{U}_q^e} \right\} 
			%%\label{eq_def_inter_sets_input}
	%\end{equation}
%%% FIN TEST
		
\noindent	with
	$\overline{\mathcal{X}}_g :=	\{ \overline{x}_0~\in~\mathbb{R}^{n + l} \ | \ \sigma_{\overline{X}_g}(\overline{x}_0)~\geq~0  \}$, 
	$\overline{\mathcal{X}}^e_h :=	\{ \overline{x}_0~\in~\mathbb{R}^{n + l} \ | \ \sigma_{\overline{X}^e_h}(x_0)~=~0 \}$,
	$\mathcal{U}_i :=	\{ u^a~\in~\mathbb{R}^{m} \ | \ \sigma_{U_i}(u^a)~\geq~0 \}$, 
	$\mathcal{U}^e_j :=	\{ u^a~\in~\mathbb{R}^{m} \ | \ \sigma_{U^e_j}(u^a)~=~0 \}$,
	%and
	$\overline{\mathcal{W}}_p :=	\{ \overline{w}~\in~\mathbb{R}^{n_v+n_w} \ | \ \sigma_{\overline{W}_p}(\overline{w})~\geq~0 \}$, 
	$\overline{\mathcal{W}}^e_q :=	\{ \overline{w}~\in~\mathbb{R}^{n_v+n_w} \ | \ \sigma_{\overline{W}^e_q}(w)~=~0 \}$,
	%and
	$\mathcal{A}_r :=	\{ a~\in~\mathbb{R}^{m_a} \ | \ \sigma_{A_r}(a)~\geq~0 \}$, 
	and
	${\mathcal{A}^e_t :=	\{ a~\in~\mathbb{R}^{m_a} \ | \ \sigma_{A^e_t}(a)~=~0 \}}$.	

	\begin{remark}
			The introduction of QCEs makes possible to exactly represent anomalies impact leading to constant~$a_k$, such as in the cases of DoS or attack by upper saturation.
	\end{remark}
	
	%\begin{remark}
		%If more natural to get the set $\mathcal{X}_0$ and $\Xi_0$, it is direct then to get $\overline{X}_0$.
		%For instance, $\sigma_{X}(x_0) \geq 0 \Leftrightarrow \sigma_{\overline{X}}(\overline{x}_0) \geq 0$ where $\overline{X} = \begin{bmatrix} I & 0 & 0 \\ 0 & 0 & 1 \end{bmatrix}^TX\begin{bmatrix} I & 0 & 0 \\ 0 & 0 & 1 \end{bmatrix}$ 
	%\end{remark}
	
	The problem can now be stated as follows.
	
	\begin{problem} \label{pb_CT_computation} \ \\
		\textsc{Given} the system (\ref{eq_SKA}) subject to input anomalies modeled by~$(\Gamma_u,\Gamma_a,\mathcal{A})$,
		\textsc{Given} a safety set~$\overline{\mathcal{S}}$, an initial set~$\overline{\mathcal{X}}_0 \subseteq \overline{\mathcal{S}}$, an input-limitation set~$\mathcal{U}$ and a deviation set~$\overline{W}$, 
		%\\
		\textsc{Given} the approximation (\ref{eq_def_inter_sets_safety})-(\ref{eq_def_inter_sets_A}) in the QC framework, \\
		 %\textsc{Given} 
		%\begin{itemize}
			%\item the system (\ref{eq_SKA}) subject to input anomalies modeled by~$(\Gamma_u,\Gamma_a,\mathcal{A})$,
			%\item a safety set~$\overline{\mathcal{S}}$, an initial set~$\mathcal{X}_0 \subseteq \overline{\mathcal{S}}$, an input-limitation set~$\mathcal{U}$ and a deviation set~$\overline{W}$, 
			%\item the approximation (\ref{eq_def_inter_sets_initial_set})-(\ref{eq_def_inter_sets_A}) in the QC framework,
		%\end{itemize}
		\textsc{Find} the maximal time-index $\boldsymbol{k_f}$ such that system (\ref{eq_SKA}) is safe over the time-window $\llbracket 0, \boldsymbol{k_f} \rrbracket$ for all initial states of~$\overline{\mathcal{X}}_0$, for all anomalies of~$\mathcal{A}$ resulting in admissible inputs of~$\mathcal{U}$, and for all elements of~$\overline{W}$:
		\begin{equation*}
					\max_{\boldsymbol{k_f} \in \mathbb{N}}{\boldsymbol{k_f}}
		\end{equation*}
		subject to $\forall \overline{x}_0 \in \overline{\mathcal{X}}_0$, $\forall k \in \llbracket 0, \boldsymbol{k_f} \rrbracket$, $\forall a_k \in \mathcal{A}$ such that $u_k^a \in \mathcal{U}$, $\forall \overline{w}_k \in \overline{W}$, $\overline{x}_k \in \overline{S}$.
	\end{problem}

\section{Critical-time computation} \label{sec_Main_result}

		%In the sequel, the sub-problem of testing if the state~$\overline{x}_{k_f}$ belongs to the safety set~$\overline{\mathcal{S}}$, over all initial states of~$\overline{\mathcal{X}}_0$, past anomalies of~$\mathcal{A}$ resulting in admissible inputs of~$\mathcal{U}$, and uncertainties of~$\overline{\mathcal{W}}$, for a fixed time-index~$k_f$ is first solved.
		%Then, an iterative LMI-based algorithm is proposed to find the maximum~$k_f$, providing then an under-estimate of the critical-time~$k_c$: $k_f \leq k_c$. 

	In this section, first the following problem is solved (Theorem~\ref{thm_CT_computation}): how to check if the state~$\overline{x}_{k_f}$ belongs to the safety set~$\overline{\mathcal{S}}$, over all initial states of~$\overline{\mathcal{X}}_0$, occurred anomalies in~$\mathcal{A}$ resulting in admissible inputs in~$\mathcal{U}$, and uncertainties in~$\overline{\mathcal{W}}$, for a fixed time-index~$k_f$. 		
	%In this section, first the following problem is solved: how to check if the state~$\overline{x}_{k_f}$ belongs to the safety set~$\overline{\mathcal{S}}$, over all initial states of~$\overline{\mathcal{X}}_0$, occurred anomalies in~$\mathcal{A}$ resulting in admissible inputs in~$\mathcal{U}$, and uncertainties in~$\overline{\mathcal{W}}$, for a fixed time-index~$k_f$. 		
	%In this section, first the following problem is solved: how to check if the state~$\overline{x}_{k_f}$ belongs to the safety set~$\overline{\mathcal{S}}$, over all initial states of~$\overline{\mathcal{X}}_0$, anomalies in~$\mathcal{A}$ resulting in admissible inputs in~$\mathcal{U}$, and uncertainties in~$\overline{\mathcal{W}}$, for a fixed~$k_f$. 		
	Then, an iterative algorithm (Algorithm~\ref{Algo_CT_computation}) is proposed to find the maximal~$\boldsymbol{k_f}$, providing an underestimate of the critical-time.

	\begin{thm} \label{thm_CT_computation} \ \\
	 Let $(\Sigma - \mathcal{K})$ be the interconnected system modeled by~(\ref{eq_SKA}).
	%with initial state $x_0 \in \mathbb{R}^n$.
		Assume that the safety set~$\overline{\mathcal{S}}$, the initial-state set~$\overline{\mathcal{X}}_0$, the  input-limitation set~$\mathcal{U}$, the deviation set~$\overline{\mathcal{W}}$ and the anomaly set~$\mathcal{A}$ are provided with QC-based approximations given by~(\ref{eq_def_inter_sets_safety})-(\ref{eq_def_inter_sets_A}).
		Assume that the time-index $k_f \in \mathbb{N}$ is fixed. \\
			Then $(i) \Rightarrow (ii)$.
			\begin{itemize}				
				%\item[$(i)$] $\forall s \in  \llbracket 1, n_{\overline{\mathcal{S}}} \rrbracket$, 		
				%$\exists \left\{\boldsymbol{\phi_{s,g}} \in \mathbb{R}^+, \boldsymbol{\varphi_{s,h}} \in \mathbb{R} \right\}_{g,h = 1}^{n_{\overline{\mathcal{X}}_0},m_{\overline{\mathcal{X}}_0}}$,	
				%$\forall k \in \llbracket 0, k_f -1 \rrbracket$, 
				%$\exists \left\{\boldsymbol{\pi_{s,k,i}} \in \mathbb{R}^+, \boldsymbol{\varpi_{s,k,j}} \in \mathbb{R} \right\}_{i,j = 1}^{n_{\mathcal{U}},m_{\mathcal{U}}}$,	
				%$\exists \left\{\boldsymbol{\rho_{s,k,p}} \in \mathbb{R}^+, \boldsymbol{\varrho_{s,k,q}} \in \mathbb{R} \right\}_{p,q = 1}^{n_{\overline{\mathcal{W}}},m_{\overline{\mathcal{W}}}}$,							
				%$\exists \left\{\boldsymbol{\sigma_{s,k,r}} \in \mathbb{R}^+, \boldsymbol{\varsigma_{s,k,t}} \in \mathbb{R} \right\}_{r,t = 1}^{n_{\mathcal{A}},m_{\mathcal{A}}}$, 
				\item[$(i)$] $\forall s \in  \llbracket 1, n_{\overline{\mathcal{S}}} \rrbracket$, 		
				$\exists \left\{\boldsymbol{\phi_{s,g}} \in \mathbb{R}^+\right\}_{g = 1}^{n_{\overline{\mathcal{X}}_0}}$,	
				$\exists \left\{\boldsymbol{\varphi_{s,h}} \in \mathbb{R} \right\}_{h = 1}^{m_{\overline{\mathcal{X}}_0}}$,	
				$\forall k \in \llbracket 0, k_f -1 \rrbracket$, 
				$\exists \left\{\boldsymbol{\pi_{s,k,i}} \in \mathbb{R}^+ \right\}_{i= 1}^{n_{\mathcal{U}}}$,	
				$\exists \left\{\boldsymbol{\varpi_{s,k,j}} \in \mathbb{R} \right\}_{j = 1}^{m_{\mathcal{U}}}$,	
				$\exists \left\{\boldsymbol{\rho_{s,k,p}} \in \mathbb{R}^+ \right\}_{p= 1}^{n_{\overline{\mathcal{W}}}}$,							
				$\exists \left\{\boldsymbol{\varrho_{s,k,q}} \in \mathbb{R} \right\}_{q = 1}^{m_{\overline{\mathcal{W}}}}$,							
				$\exists \left\{\boldsymbol{\sigma_{s,k,r}} \in \mathbb{R}^+\right\}_{r = 1}^{n_{\mathcal{A}}}$, 				
				$\exists \left\{\boldsymbol{\varsigma_{s,k,t}} \in \mathbb{R} \right\}_{t = 1}^{m_{\mathcal{A}}}$, 				
				such that:
				\end{itemize}	
					\begin{align}				
						&M_{\overline{x}_{k_f}}^T \overline{S}_s M_{\overline{x}_{k_f}} - M_{\overline{x}_0}^T \left( \sum_{g=1}^{n_{\overline{\mathcal{X}}_0}}{\boldsymbol{\phi_{s,g}} \overline{X}_g } + \sum_{h=1}^{m_{\overline{\mathcal{X}}_0}}{\boldsymbol{\varphi_{s,h}} \overline{X}_h^e }  \right) M_{\overline{x}_0}
						\nonumber \\
						-	&\sum_{k=0}^{k_f - 1}{N_{u_k^a}^T \left( \sum_{i=1}^{n_{\mathcal{U}}}{ \boldsymbol{\pi_{s,k,i}} U_i} + \sum_{j=1}^{m_{\mathcal{U}}}{ \boldsymbol{\varpi_{s,k,j}} U_j^e} \right) N_{u_k^a}} \nonumber \\
						-
							&\sum_{k=0}^{k_f - 1}{E_k^T \left( \sum_{p=1}^{n_{\overline{\mathcal{W}}}}{ \boldsymbol{\rho_{s,k,p}} \overline{W}_p} + \sum_{q=1}^{m_{\overline{\mathcal{W}}}}{ \boldsymbol{\varrho_{s,k,q}} \overline{W}_q^e} \right) E_k}
							%\geq 0 
							\nonumber \\
						-
							&\sum_{k=0}^{k_f - 1}{F_k^T \left( \sum_{r=1}^{n_{\mathcal{A}}}{ \boldsymbol{\sigma_{s,k,r}} A_r} + \sum_{t=1}^{m_{\mathcal{A}}}{ \boldsymbol{\varsigma_{s,k,t}} A_t^e} \right) F_k}
							\geq 0
							\label{eq_cond_LMI}
					\end{align}	
					with 
					\begin{align*}
							M_{\overline{x}_k} &= \begin{bmatrix} A^k & P_k & Q_k & 0 \\ 0 & 0 & 0 & 1 \end{bmatrix} \\
							P_k &= \begin{bmatrix} A^{k-1} W & \ldots  & A W &  W & 0_{(n+l)\times (k_f-k)(n_v+n_w)}  \end{bmatrix} \\
							Q_k &= \begin{bmatrix} A^{k-1} B_a &  \ldots & A B_a &  B_a & 0_{(n+l)\times (k_f-k)m_a} \end{bmatrix} \\
							N_{u_k^a} &= C_u M_{\overline{x}_k} + V_uE_k + \Gamma_a F_k
					\end{align*}
					where $P_0 = 0$ and	$Q_0 = 0$,							
					and $E_k$, $F_k$ are such that: 	
					\begin{equation*}
							E_k z_{k_f}	
												= \begin{bmatrix} \overline{w}_k \\ 1 \end{bmatrix}
							\qquad
							F_k z_{k_f}	
												= \begin{bmatrix} a_k \\ 1 \end{bmatrix}												
								%\label{eq_def_Ek_Fk}
					\end{equation*}						
					where $z_{k_f}^T := \begin{bmatrix}   \overline{x}_0^T & \overline{w}_{\llbracket 0, k_f-1 \rrbracket}^T & a_{\llbracket 0, k_f-1 \rrbracket}^T & 1 \end{bmatrix}^T$.
					%%\begin{equation*}
							%$E_k z_{k_f}	
												%= \begin{bmatrix} \overline{w}_k \\ 1 \end{bmatrix}$,
							%$F_k z_{k_f}	
												%= \begin{bmatrix} a_k \\ 1 \end{bmatrix}$,												
								%%\label{eq_def_Ek_Fk}
					%%\end{equation*}						
					%where $z_{k_f}^T := \begin{bmatrix}   \overline{x}_0^T & \overline{w}_{\llbracket 0, k_f-1 \rrbracket}^T & a_{\llbracket 0, k_f-1 \rrbracket}^T & 1 \end{bmatrix}^T$.					
			\begin{itemize}
				%\item[$(ii)$] The system is safe at time~$k_f$ $\overline{x}_{k_f} \in \mathcal{S}$ for all initial states~$\overline{x}_0 \in \overline{\mathcal{X}}_0$, past deviation elements of~$\overline{\mathcal{W}}$ and anomalies of~$\mathcal{A}$ resulting in admissible inputs $u_k^a \in \mathcal{U}$.
				\item[$(ii)$] The system is safe at time-index~$k_f$: $\overline{x}_{k_f} \in \mathcal{S}$ for all initial states~$\overline{x}_0 \in \overline{\mathcal{X}}_0$, for all anomalies $a_k \in \mathcal{A}$ resulting in admissible inputs $u_k^a \in \mathcal{U}$, and for all deviation elements ~$\overline{w}_k \in \overline{\mathcal{W}}$.
				%$\forall x_0 \in \mathcal{X}_0$,
				%$\forall u_{a_{\llbracket 0, k_f-1 \rrbracket}} \in \mathcal{U}^{k_f}$, 
				%$\forall w_{\llbracket 0, k_f-1 \rrbracket} \in \mathcal{W}^{k_f}$,
				%$x_{k_f} \in \mathcal{S}$.
			\end{itemize}
	\end{thm}

	\begin{proof}
			Using the recursive equation (\ref{eq_SKA}), it comes	
		\begin{equation*}
			M_{\overline{x}_{k_f}} z_{k_f}	 = \begin{bmatrix} \overline{x}_{k_f} \\ 1 \end{bmatrix}
			\quad
			M_{\overline{x}_0}z_{k_f} = \begin{bmatrix} \overline{x}_{0} \\ 1 \end{bmatrix}
			\quad
				N_{u_k^a} z_{k_f}	
												= \begin{bmatrix} u_k^a \\ 1 \end{bmatrix}	
		\end{equation*}			
				Then, by applying Lemma~\ref{lem_Sprocedure} to~(\ref{eq_cond_LMI}),
				%and using~(\ref{eq_xkc_intermsof_x0_ua}) and~(\ref{eq_def_Ek_Fk}), 
				%it comes that
				$\sigma_{\overline{S}_s}(\overline{x}_f) \geq 0$ holds for all $\overline{x}_0$, $ \left\{ \overline{w}_k \right\}_{k=0}^{k_f-1}$, and $\left\{ u_{k}^a \right\}_{k=0}^{k_f-1}$ such that:
			\begin{align*}
					&\forall  g \in \llbracket 1, n_{\overline{\mathcal{X}}_0} \rrbracket, \ \sigma_{\overline{X}_g}(\overline{x}_0) \geq 0						
					\quad
					\forall  h \in \llbracket 1, m_{\overline{\mathcal{X}}_0} \rrbracket, \ \sigma_{\overline{X}_h^e}(\overline{x}_0) = 0				
			%\end{equation*}				
			%\begin{equation*}
			\\
					&\forall  i \in \llbracket 1, n_{\mathcal{U}} \rrbracket, \ \sigma_{U_i}(u_{k}^a) \geq 0						
				 \	\quad
					\forall  j \in \llbracket 1, m_{\mathcal{U}} \rrbracket, \ \sigma_{U_j^e}(u_{k}^a) = 0				
			%\end{equation*}
			%\begin{equation*}
				\\
					&\forall  p \in \llbracket 1, n_{\overline{\mathcal{W}}} \rrbracket, \ \sigma_{\overline{W}_p}(\overline{w}_k) \geq 0
											\quad
					\forall  q \in \llbracket 1, m_{\overline{\mathcal{W}}} \rrbracket, \ \sigma_{\overline{W}_q^e}(\overline{w}_k) = 0		
					\\
					&\forall  r \in \llbracket 1, n_{\mathcal{A}} \rrbracket, \ \sigma_{A_r}(a_k) \geq 0						
				 \	\quad
					\forall  t \in \llbracket 1, m_{\mathcal{A}} \rrbracket, \ \sigma_{A_t^e}(a_k) = 0							
			\end{align*}					
		%Thus, $x_{k_f} \in \bigcap_{i=1}^{n_{\mathcal{S}}}{\mathcal{S}_i}$ for all $\left\{ u_{a_k} \right\}_{k=0}^{k_f-1}$ 
		%and~$ \left\{ w_k \right\}_{k=0}^{k_f-1}$ 
		%such that $u_{a_k} \in \left\{ \cap_{j=1}^{n_{\mathcal{U}}}{\ \mathcal{U}_j} \right\} \bigcap \left\{ \cap_{l=1}^{m_{\mathcal{U}}}{\ \mathcal{U}_l^e} \right\}$
		%and 
		%%for all $ \left\{ w_k \right\}_{k=0}^{k_f-1}$ such that 
		%$w_k \in \left\{ \cap_{f=1}^{n_{\mathcal{W}}}{\ \mathcal{W}_f} \right\} \bigcap \left\{ \cap_{g=1}^{m_{\mathcal{W}}}{\ \mathcal{W}_g^e} \right\}$.
		%Thus, $x_{k_f} \in \bigcap_{i=1}^{n_{\mathcal{S}}}{\mathcal{S}_i}$ for all $\left\{ u_{a_k} \right\}_{k=0}^{k_f-1}$ 
		%such that $u_{a_k} \in \left\{ \cap_{j=1}^{n_{\mathcal{U}}}{\ \mathcal{U}_j} \right\} \bigcap \left\{ \cap_{l=1}^{m_{\mathcal{U}}}{\ \mathcal{U}_l^e} \right\}$
		%and for all $ \left\{ w_k \right\}_{k=0}^{k_f-1}$ 
		%such that
		%%for all $ \left\{ w_k \right\}_{k=0}^{k_f-1}$ such that 
		%$w_k \in \left\{ \cap_{f=1}^{n_{\mathcal{W}}}{\ \mathcal{W}_f} \right\} \bigcap \left\{ \cap_{g=1}^{m_{\mathcal{W}}}{\ \mathcal{W}_g^e} \right\}$.	
		Thus, 
		%$x_{k_f} \in \bigcap_{s=1}^{n_{\mathcal{S}}}{\mathcal{S}_s}$ 
		for all $\overline{x}_0 \in \left\{ \cap_{g=1}^{n_{\overline{\mathcal{X}}_0}}{\ \overline{\mathcal{X}}_g} \right\} \bigcap \left\{ \cap_{h=1}^{m_{\overline{\mathcal{X}}_0}}{\ \overline{\mathcal{X}}_h^e} \right\}$,
		for all $\left\{ a_k\right\}_{k=0}^{k_f-1}$ such that $a_k \in \left\{ \cap_{r=1}^{n_{\mathcal{A}}}{\ \mathcal{A}_r} \right\} \bigcap \left\{ \cap_{t=1}^{m_{\mathcal{A}}}{\ \mathcal{A}_t^e} \right\}$ and $u_{k}^a \in \left\{ \cap_{i=1}^{n_{\mathcal{U}}}{\ \mathcal{U}_i} \right\} \bigcap \left\{ \cap_{j=1}^{m_{\mathcal{U}}}{\ \mathcal{U}_j^e} \right\}$, 
		and 				 for all $ \left\{ \overline{w}_k \right\}_{k=0}^{k_f-1}$ such that
		$\overline{w}_k \in \left\{ \cap_{p=1}^{n_{\overline{\mathcal{W}}}}{\ \overline{\mathcal{W}}_p} \right\} \bigcap \left\{ \cap_{q=1}^{m_{\overline{\mathcal{W}}}}{\ \overline{\mathcal{W}}_q^e} \right\}$,
		one has $\overline{x}_{k_f} \in \bigcap_{s=1}^{n_{\overline{\mathcal{S}}}}{\overline{\mathcal{S}}_s}$.
	 %The conclusion comes then by (\ref{eq_def_inter_sets_safety})-(\ref{eq_def_inter_sets_dev}).
	 %The conclusion is obtained by~{(\ref{eq_def_inter_sets_safety})-(\ref{eq_def_inter_sets_dev})}.
	 The inclusions~{(\ref{eq_def_inter_sets_safety})-(\ref{eq_def_inter_sets_A})} lead to the conclusion.
	\end{proof}

	\begin{remark}
	Condition~$(i)$ is a feasibility optimization problem with LMI constraints. It is then solvable in polynomial-time.
	%This computational advantage comes at the price of the conservatism of condition~$(i)$, that implies but is not equivalent to condition~$(ii)$. This results from the application of the S-procedure (see \cite{IwH:05,Jon:01b}, and the references therein, for a detailed discussion).
	%This computational advantage comes at the price of the conservatism of condition~$(i)$, that is only sufficient. This results from the application of the S-procedure (see \cite{IwH:05,Jon:01b}, and the references therein, for a detailed discussion).
	This computational advantage comes at the price of the conservatism of condition~$(i)$, that is only sufficient. This results from the application of the S-procedure (see \cite{Jon:01b}, and the references therein, for a detailed discussion).
	\end{remark}
	
		%\subsection{Critical-time computation}	
		
				%Then, an iterative LMI-based algorithm is proposed to find the maximum~$k_f$, providing then an under-estimate of the critical-time~$k_c$: $k_f \leq k_c$. 
		%
	%%In order to compute the critical-time~$k_c$, Algorithm~\ref{Algo_CT_computation} is proposed. Based on the sufficient condition provided by Lemma~\ref{lem_CT_computation_fixed_x0}, the time index~$k_f$ is incremented until there is no more solution to the LMI feasibility problem. An underestimate of the critical-time~$k_c$ is so obtained and ensures that the system is safe at least until the obtained result~$k_f^*$.
	%In order to compute the critical-time~$k_c$, Algorithm~\ref{Algo_CT_computation} is proposed. Based on the sufficient condition provided by Theorem~\ref{thm_CT_computation}, the time index~$k_f$ is incremented until no solution is found to the LMI feasibility problem. An underestimate of the critical-time~$k_c$ is so obtained and ensures that the system is safe at least until the obtained result~$k_f^*$.

		In order to compute the critical-time~$k_c$, the following iterative LMI-based algorithm (Algorithm~\ref{Algo_CT_computation}) is proposed: based on the sufficient condition provided by Theorem~\ref{thm_CT_computation}, the time index~$k_f$ is incremented until no solution is found to the LMI feasibility problem. 
		%In order to compute the critical-time~$k_c$, the following iterative LMI-based algorithm is proposed: based on the sufficient condition~$(i)$ of Theorem~\ref{thm_CT_computation}, the time index~$k_f$ is incremented until no solution is found to the LMI feasibility problem. 
		An under-estimate~$k_f^*$ of the critical-time~$k_c$ is so obtained~$k_f^* \leq k_c$ and ensures that the system is safe at least until the obtained result~$k_f^*$.

	%In order to compute the critical-time~$k_c$, Algorithm~\ref{Algo_CT_computation} is proposed. Based on the sufficient condition provided by Lemma~\ref{thm_CT_computation}, the time index~$k_f$ is incremented until no solution is found to the LMI feasibility problem. An underestimate of the critical-time~$k_c$ is so obtained and ensures that the system is safe at least until the obtained result~$k_f^*$.
	
\begin{algorithm}[!htb]
\DontPrintSemicolon
	$k_f \leftarrow 0$ \\
	$t \leftarrow true$ \\
   \While{$t$}
   {
			$k_f \leftarrow k_f + 1$ \\
			$t \leftarrow isFeasible$(  	
			 $\forall s \in  \llbracket 1, n_{\overline{\mathcal{S}}} \rrbracket$, 		
				$\exists \left\{\boldsymbol{\phi_{s,g}} \in \mathbb{R}^+\right\}_{g = 1}^{n_{\overline{\mathcal{X}}_0}}$,	
				$\exists \left\{\boldsymbol{\varphi_{s,h}} \in \mathbb{R} \right\}_{h = 1}^{m_{\overline{\mathcal{X}}_0}}$,	
				$\forall k \in \llbracket 0, k_f -1 \rrbracket$, 
				$\exists \left\{\boldsymbol{\pi_{s,k,i}} \in \mathbb{R}^+ \right\}_{i= 1}^{n_{\mathcal{U}}}$,	
				$\exists \left\{\boldsymbol{\varpi_{s,k,j}} \in \mathbb{R} \right\}_{j = 1}^{m_{\mathcal{U}}}$,	
				$\exists \left\{\boldsymbol{\rho_{s,k,p}} \in \mathbb{R}^+ \right\}_{p= 1}^{n_{\overline{\mathcal{W}}}}$,							
				$\exists \left\{\boldsymbol{\varrho_{s,k,q}} \in \mathbb{R} \right\}_{q = 1}^{m_{\overline{\mathcal{W}}}}$,							
				$\exists \left\{\boldsymbol{\sigma_{s,k,r}} \in \mathbb{R}^+\right\}_{r = 1}^{n_{\mathcal{A}}}$, 				
				$\exists \left\{\boldsymbol{\varsigma_{s,k,t}} \in \mathbb{R} \right\}_{t = 1}^{m_{\mathcal{A}}}$, 	
%
				%$\forall k \in \llbracket 0, k_f -1 \rrbracket$, $\forall i \in  \llbracket 1, n_{\mathcal{S}} \rrbracket$, \\[0.1em]
				%$\exists \left\{\boldsymbol{\tau_{i,j,k}} \in \mathbb{R}^+ \right\}_{j = 1}^{n_{\mathcal{U}}}$, 
				%$\exists \left\{\boldsymbol{\mu_{i,l,k}} \in \mathbb{R} \right\}_{l=1}^{m_{\mathcal{U}}}$, \\[0.1em]
				%$\exists \left\{\boldsymbol{\rho_{i,f,k}} \in \mathbb{R}^+ \right\}_{f = 1}^{n_{\mathcal{W}}}$, 
				 %$\exists \left\{ \boldsymbol{\nu_{i,g,k}} \in \mathbb{R} \right\}_{g=1}^{m_{\mathcal{W}}}$, \\[0.1em]
				such that (\ref{eq_cond_LMI}) holds.
			) 
			\\
			/* true if a solution is found /*
   }	
	$k_f \leftarrow k_f-1$ \\
	\textbf{return} $k_f$
	%\tcc{ }
\nonl  \
\caption{Critical-time computation}
\label{Algo_CT_computation}
% cf http://shantoroy.com/latex/how-to-write-algorithm-in-latex/
\end{algorithm}		
%
%

%\input{Files/5_Illustration.tex}

%\section{Illustration} \label{sec_illustration} 
\section{Case study - the quadruple-tank system} \label{sec_illustration} 

%$\rightarrow$ quel intérêt de l'approche dans le cas d'un DoS par rapport à une simulation ? La prise en compte des incertitudes ! Donc faire comparaison avec du Monte Carlo ?

		%\begin{itemize}
			%\item 3 reasons for which computed CT is under NL simulation
				%\begin{enumerate}
					%\item Conservatism of the approach
					%\item Approximation with an uncertain model
					%\item We test for all $u_a \in \mathcal{U}$ in computation and only for $u_a \in \{u_{min}, u_{max}\} \subseteq \mathcal{U}$ in simulation $\Rightarrow$ we guarantee more than simply simulation ! 
				%\end{enumerate}
		%\end{itemize}

		\begin{figure}[!htb]
      \centering
      \includegraphics[width=0.75\columnwidth]{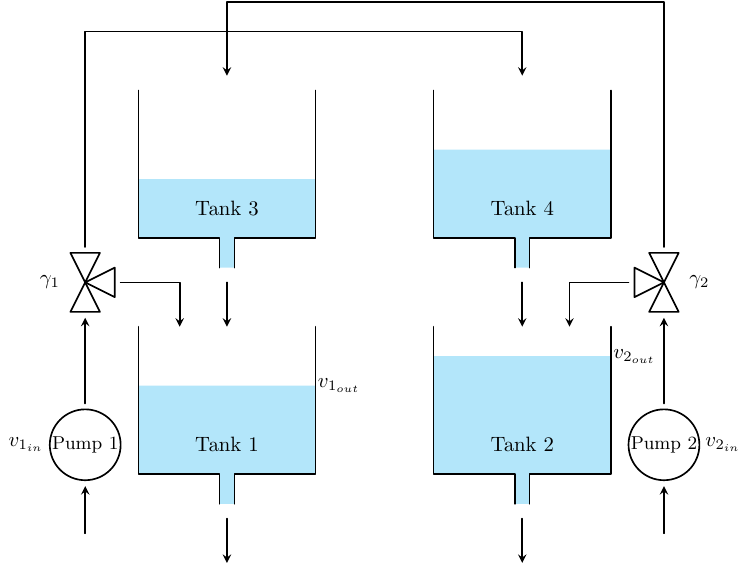}
      \caption{The quadruple-tank system}
      \label{fig_4_water_tanks}
   \end{figure}	
	
		The quadruple-tank system (Fig.~\ref{fig_4_water_tanks}) originally described in~\cite{Joh:00}, is considered to illustrate the benefits of the critical-time.
		%This system has received particular attention in the CPS security literature for simulating cyber-attacks on CPS and evaluating their physical impacts (see for instance~\cite{TSSJ:15}). 
		This benchmark has been used in the CPS security literature for simulating cyber-attacks, evaluating their physical impacts and developing counter-measures (see~\cite{TSSJ:15,SRPVQ:19} for instance). 
		It consists of four tanks with bottom outlet holes, arranged such that tank~3 is above	tank~1 and tank~4 above tank~2.
		The inputs are the voltages~$v_{1_{in}}$, $v_{2_{in}}$ applied to two pumps, and the outputs are the voltages~$v_{1_{out}}$, $v_{2_{out}}$ from level measurement devices of tanks~1 and~2.
		%The input liquid from pump~1 (resp. pump~2) is split into tank~$1$ and~$4$ (resp. tank~$2$ and~$3$) using a valve, where~$\gamma_1 \in (0,1)$ and~$\gamma_2 \in (0,1)$ are the valves coefficients.
		The input liquid from pump~1 (resp. pump~2) is split into tank~$1$ and~$4$ (tank~$2$ and~$3$) using a valve with coefficient~$\gamma_1~\in~(0,1)$ (resp. $\gamma_2 \in (0,1)$).
		The system is modeled as follows:
		\begin{equation}
				\left\{ 
						\begin{aligned}
							\dot{h}(t) &= A_P \sqrt{h(t)} + B_P(\gamma_1,\gamma_2) v_{in}(t) \\
							v_{out}(t) &= C_P h(t) + d(t)
						\end{aligned}
				\right.
				\label{eq_QTP}
		\end{equation}
		%
		%$\rightarrow$ ENLEVER $d(t)$ ==> dire que c'est le modèle du procédé puis rajouter après pour prendre en compte les bruits de capteur
		%
		where~$h = \begin{bmatrix} h_1 & h_2 & h_3 & h_4 \end{bmatrix}^T$, $h_i$ the liquid level of tank~$i$,
		%$d$ the measurement noise,
		$v_{in} = \begin{bmatrix} v_{1_{in}} & v_{2_{in}} \end{bmatrix}^T$, $v_{out} = \begin{bmatrix} v_{1_{out}} & v_{2_{out}} \end{bmatrix}^T$ and~$d$ the sensor noise.
		The matrices~$A_P$ and~$C_P$ depend on constant physical parameters of the system, while~$B_P(\gamma_1,\gamma_2)$ is an affine matrix-valued function.
		See~\cite{Joh:00} for numerical values.
		%\begin{equation*}
				%A_P := \sqrt{2g} \begin{bmatrix} 
									%- \frac{a_1}{A_1} & 0 & \frac{a_3}{A_1} & 0 \\
									%0 & - \frac{a_2}{A_2} & 0 & \frac{a_4}{A_2} \\
									%0 & 0 & - \frac{a_3}{A_3} & 0 \\
									%0 & 0 & 0 &  - \frac{a_4}{A_4}
							%\end{bmatrix}
		%\end{equation*}
		%\begin{equation*}
				%B_P :=  \begin{bmatrix} 
										%\gamma_1 \frac{k_1}{A_1} & 0 \\ 
										%0 & \gamma_2 \frac{k_2}{A_2} \\
										%0 & (1 - \gamma_2) \frac{k_2}{A_3} \\
										%(1 - \gamma_1) \frac{k_1}{A_4}
								%\end{bmatrix}
				%\
				%C_P := \begin{bmatrix} k_c & 0 & 0 & 0 \\ 0 & k_c & 0 & 0 \end{bmatrix}
		%\end{equation*}

		%
		Due to physical limitations of the pumps, the inputs $v_{1_{in}}$ and $v_{2_{in}}$ belong to $[v_{min},v_{max}] = [0 \text{ V}, 10 \text{ V}]$.
		Moreover, for safety reasons, the levels~$h_1$ and~$h_2$ are constrained to belong to~$[h_{12_{min}},h_{12_{max}}] = [9 \text{ cm}, 20 \text{ cm}]$ and the levels~$h_3$ and~$h_4$ to~$[h_{34_{min}},h_{34_{max}}] = [0.1 \text{ cm}, 5 \text{ cm}]$.		
		In particular, this prevents the tanks to overflow or to be empty, and keeps a minimum output flow for tanks~$1$ and~$2$.		
		
		%The process is assumed to be remotely controlled, where the control objective is to set~$h_1$ and~$h_2$ to given set-points~$h_1^0 = 12.4 \text{ cm}$ and ~$h_2^0 = 12.7 \text{ cm}$.
		The system is remotely controlled through a communication network by PI controllers, that compute the control input~$v_c$ from a reference signal~$r$ and measurement $v_{out}$ such as~$v_{i_c} = K_i(1 + \frac{1}{T_i \cdot s} )(r_i - v_{i_{out}} )$ where~$K_1 = 3$, $K_2 = 2.7$, $T_1 = 30s$, $T_2 = 40s$.		
		%The control objective is to set~$h_1$ and~$h_2$ to set-points~$h_1^0 = 12.4 \text{ cm}$ and ~$h_2^0 = 12.7 \text{ cm}$.
		%The control objective is to set~$h_1$ and~$h_2$ to set-points~$h_1^0 = 12.4 \text{ cm}$ and~$h_2^0 = 12.7 \text{ cm}$.
		The control objective is to set~$h_1$ and~$h_2$ to~$h_1^0 = 12.4 \text{ cm}$ and~$h_2^0 = 12.7 \text{ cm}$.
		%It is assumed that anomalies appear when the process is at a stationary operating point~$(h^0,v_{in}^0)$.
		Anomalies are assumed to appear when the process is at a stationary operating point~$(h^0,v_{in}^0)$.
		It is known~\cite{Joh:00} that, given~$(h_1^0, h_2^0)$, the values of~$h_3^0$, $h_4^0$ and~$v_{in}^0$ are fully determined by the valves coefficients~$(\gamma_1, \gamma_2)$.
		%provided that~$\gamma_1 + \gamma_2 \neq 1$.
		%This implies that, given~$h_1^0$ and~$h_2^0$, the values of~$h_3^0$, $h_4^0$ and~$v_{in}^0$
		The choice of~$(\gamma_1, \gamma_2)$ impacts then both the operating point and the system dynamics.
		Thus, the aim is to study the influence of~$(\gamma_1, \gamma_2)$ on the critical-time.
		In addition to evaluate the time allowed to defense mechanism given a couple~$(\gamma_1, \gamma_2)$, this provides a means to select couples with high critical-time.
		%preventing then potential damages of certain anomalies. % as being detected and removed sufficiently early
		%To achieve this, the critical-time is computed for a finite number of couple~$(\gamma_1, \gamma_2)$, which are arbitrarily chosen such that~$\gamma_1 \in [0.65, 0.75]$ and~$\gamma_2 \in [0.55, 0.65]$.
		To this end, the critical-time is computed for a finite number of arbitrarily-chosen~$\gamma_1 \in [0.65, 0.75]$, $\gamma_2 \in [0.55, 0.65]$.
	
	%\begin{remark}
				%For the selected values of~$(\gamma_1,\gamma_2)$, the process is minimum-phase, which prevents zero-dynamics attacks~\cite{TSSJ:15b}. 
	%\end{remark}
		
		%To achieve this,
		%In the sequel, an uncertain discrete-time approximation of the process is first computed. Then, Algorithm~\ref{Algo_CT_computation} is applied and the obtained result is discussed. Finally, simulations of particular anomalies are made for comparison purposes.

		%$\gamma_1 \in [0.65, 0.75]$
		%$\gamma_2 \in [0.55, 0.65]$
		%*
		%+ DIRE qu'on DISCRETIZE les $\gamma$ (arbitrarily chosen) 
		%+ CELA IMPLIQUE MINIMUM PHASE + RESILIENCE TO zero-dynamics attack (Mieux: ineffective zero-dynamics attack)
		%
		%+ ASSUMPTION: HAS A DEFENSE MECHANISM THAT CAN DETECT AND REMOVE ANOMALIES UNDER CONCERN)
		%+ OR (BETTER): ONE HAS TO DEVELOP A DEFENSE STRATEGY

		%\subsection{Stationary operating point and uncertain LTI approximation}
		\subsection{Uncertain linear discrete-time approximation}
		
	For each couple~$(\gamma_1,\gamma_2)$, the corresponding model~(\ref{eq_QTP}) is linearized around the associated operating point~$(h^0, v_{in}^0)$.
				%Then, using the zero-order hold method with a sampling period of~$1 \text{s}$, an uncertain linear discrete-time model~(\ref{eq_SKA}) is obtained, where~$\overline{w}_k$ represents the deviation between~(\ref{eq_QTP}) and the approximation, especially resulting from linearization-and-discretization errors and sensor noise.
				Then, using the zero-order hold method with a sampling period of~$1 \text{s}$, the model~(\ref{eq_SKA}) is obtained, where~$\overline{w}_k$ represents the deviation between~(\ref{eq_QTP}) and the approximation, especially resulting from linearization-and-discretization errors and sensor noise.
				%Based on several simulations and worst-case scenarios, the set~$\overline{\mathcal{W}}$ is over-approximated (see Fig.~\ref{fig_W1_w1w3_Intersect_Ellips_Rectangle_with_grid}) by the intersection of two degenerated ellipsoids~$\overline{\mathcal{W}}_1$ and~$\overline{\mathcal{W}}_2$, given by
				%Based on several simulations and worst-case scenarios, the set~$\overline{\mathcal{W}}$ is over-approximated (see Fig.~\ref{fig_W1_w1w3_Intersect_Ellips_Rectangle_with_grid}) by the intersection of two degenerated ellipsoids~$\overline{\mathcal{W}}_1$ and~$\overline{\mathcal{W}}_2$, where by
				Based on several simulations, the set~$\overline{\mathcal{W}}$ is over-approximated (see Fig.~\ref{fig_W1_w1w3_Intersect_Ellips_Rectangle_with_grid}) by the intersection of two degenerated ellipsoids~$\overline{\mathcal{W}}_1$ and~$\overline{\mathcal{W}}_2$, 
				%where by
\begin{equation*}
				\overline{W}_1 = \begin{bmatrix} Q_{\overline{W}_1} & 0 & s_{\overline{W}_1} \\ 0 & 0 & 0 \\ s_{\overline{W}_1}^T & 0 & r_{\overline{W}_1}  \end{bmatrix}
		\qquad
				\overline{W}_2 = \begin{bmatrix} Q_{\overline{W}_2} & 0 & s_{\overline{W}_2} \\ 0 & 0 & 0 \\ s_{\overline{W}_2}^T & 0 & r_{\overline{W}_2}  \end{bmatrix}
		\end{equation*}
		where
		\begin{equation*}
				Q_{\overline{W}_1} = \begin{bmatrix}
															 -0.567  	& 0 	& -0.489 &  0  \\
																	0     & 0   &  0     &  0    \\
															-0.489    & 0   & -0.448 &  0  \\
																	0     & 0   &  0     &  0  
															\end{bmatrix}
					\quad
					s_{\overline{W}_1} = \begin{bmatrix} 0.000263 \\ 0 \\ 0.000433 \\ 0    \end{bmatrix}
					\quad
					r_{\overline{W}_1} = -1.20 \cdot 10^{-8}
		\end{equation*}
		\begin{equation*}
				Q_{\overline{W}_2} = \begin{bmatrix}
																	0     & 0   &  0     &  0      \\					
																	0     & -0.611 & 0   & -0.480  \\
																	0     & 0   &  0     &  0     \\	
																	0     & -0.480 & 0   & -0.409  \\	
															\end{bmatrix}			
					\quad
					s_{\overline{W}_2} = \begin{bmatrix} 0 \\ 0.000569 \\ 0 \\ 0.000501    \end{bmatrix}			
					\quad
					r_{\overline{W}_2} = -3.70 \cdot 10^{-7}
		\end{equation*}				
				%$r_{\overline{W}_1} = -1.20 \cdot 10^{-8}$, $r_{\overline{W}_2} = -3.70 \cdot 10^{-7}$, 
				and a hyper-rectangle 
				\begin{equation*}
						\left\{ \overline{w} \in \mathbb{R}^6, \ \overline{w}_{min}  \leq \overline{w} \leq \overline{w}_{max}  \right\}
				\end{equation*}
			with
			\begin{align*}
					\overline{w}_{min} &= - \begin{bmatrix} 0.013 & 0.0027 & 0 &  0 &  0.137 &  0.126 \end{bmatrix}^T \\
					\overline{w}_{max} &= \begin{bmatrix}  0.00011 &  0.0013 & 0.014 & 0.0045 & 0.143 & 0.147 \end{bmatrix}^T
			\end{align*}		 
			Similarly, the initial set~$\overline{\mathcal{X}}_0$ is over-approximated by the hyper-rectangle
			\begin{equation*}
						\left\{ \overline{x}_0 \in \mathbb{R}^6, \ \overline{x}_{0_{min}}  \leq \overline{x}_0 \leq \overline{x}_{0_{max}}  \right\}
			\end{equation*}
			with
			\begin{align*}
					\overline{x}_{0_{min}} &= - \begin{bmatrix} 0.065 & 0.047 & 0.032 &  0.028 &  0.066 &  0.061 \end{bmatrix}^T \\
					\overline{x}_{0_{max}} &= \begin{bmatrix}  0.066 &  0.056 & 0.044 & 0.024 & 0.068 & 0.077 \end{bmatrix}^T
			\end{align*}				

		\begin{figure}[!htb]
      \centering
      \includegraphics[width=0.75\columnwidth]{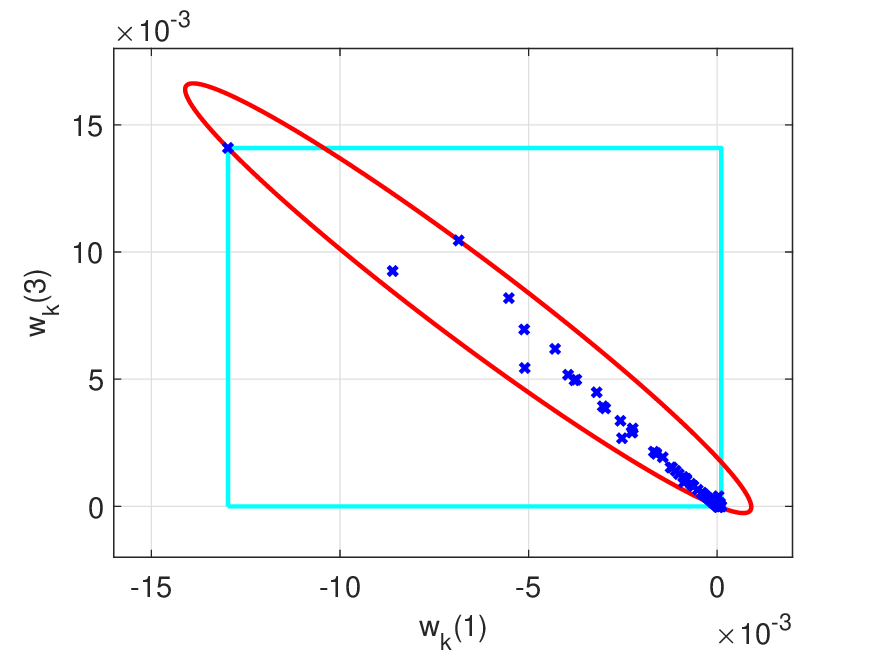}
      \caption{Projection of the hyper-rectangle (cyan solid line) and a degenerated ellipsoid (red solid line) in the plane~$(\overline{w}_k(1), \overline{w}_k(3))$. 
							%Caption figure : note that this captures the natural/coupling between the levels of tanks 1 and 3 (and respectively for tanks 2 and 4).
			%Blue crosses denote samples coming from simulations with worst-case scenarios. 
			Blue crosses denote samples coming from simulations. 
			The intersection of the rectangle and the ellipse includes the whole set of crosses.}
      \label{fig_W1_w1w3_Intersect_Ellips_Rectangle_with_grid}
   \end{figure}

		%\subsection{Critical-time computation and simulation comparison}
		\subsection{Critical-time computation and usage}

			Algorithm~\ref{Algo_CT_computation} is now applied to estimate the critical-time in three scenarios.
			The computation is performed using \mbox{\textsc{matlab}} R2022b, using the semidefinite mode of the \textsc{cvx} package~\cite{GrB:20} with the solver \textsc{mosek}~\cite{Mos:21}, on a computer with 4-core 3.00 GHz CPU and 32GB RAM.

		\subsubsection{Anomaly-specific scenario}

				A DoS anomaly, resulting for instance from an unreliable communication network or a jamming cyber-attack, is first considered.
				When an anomaly is especially risky, an anomaly-specific defense strategy may be designed in the treatment stage.
				As a case in point, a DoS may be detected using acknowledgment-based communication protocol such as TCP~\cite{SSFPS:07}.
				This anomaly is modeled as~$(\Gamma_u,\Gamma_a,\mathcal{A}) = (0_2,I_2,\mathcal{U}_{\gamma})$ where 
				\begin{equation*}
				\mathcal{U}_{\gamma} = \left\{ u \in \mathbb{R}^{2}, \ u = \begin{bmatrix} u_{\gamma_1} \\ u_{\gamma_2} \end{bmatrix}, \ u_{\gamma_i} = 0 - v_{0_i}  \right\}
				\end{equation*}
				recalling that the values of~$v_0$ depends on~$(\gamma_1,\gamma_2)$.
				%The result is displayed in Table~\ref{tab_CT_by_Alg_DoS} together with the result obtained from simulation on the non-linear model~(\ref{eq_QTP}).
				The result is displayed in Table~\ref{tab_CT_by_Alg_DoS}.
				Similarly, an attack-by-upper saturation is considered and the estimated critical-time is displayed in Table~\ref{tab_CT_by_Alg_Uppersat}. 
				Both results are compared with the result obtained from simulation on the non-linear model~(\ref{eq_QTP}).
				%It can be noticed that both result are close, arguing for little conservatism of the computing approach developed in this paper.
				It can be especially observed that results from Algorithm~\ref{Algo_CT_computation} and from simulation are close, suggesting a moderate conservatism of the computing approach developed in this paper.
	%Notice that this difference may also by induced by the tightness of the approximation and the fact that, in Algorithm~1, the whole set~$\mathcal{U}$ is tested and not only two instances.

\begin{table}[!htb]
%\caption{Critical-time (s) - DoS Alg.~\ref{Algo_CT_computation} (diff. with simulation)}
\caption{Critical-time (s) - DoS (diff. simulation vs algorithm)}
\label{tab_CT_by_Alg_DoS}
\centering
\begin{tabular}{c|ccccc}
\toprule
	\diagbox{$\gamma_1$}{$\gamma_2$}	   & $0.55$      & $0.575$ & $0.60$ & $0.625$ & $0.65$  \\ \midrule
  $0.65$  & $13$ 				& $12$         & $10$ $(+1)$ & $8$ $(+1)$  & $6$ $(+1)$ \\
  $0.675$ & $13$ $(+1)$ & $12$ $(+1)$  & $11$ $(+1)$ & $10$ $(+1)$ & $8$ $(+1)$ \\ 
	$0.70$  & $14$ $(+1)$ & $13$         & $12$        & $11$ $(+1)$ & $9$  $(+2)$ \\
  $0.725$ & $10$ $(+2)$ & $13$ $(+1)$  & $12$ $(+1)$ & $12$        & $10$ $(+1)$ \\ 
  $0.75$  & $7$ $(+2)$  & $10$ $(+1)$  & $11$ $(+2)$ & $12$        & $11$  $(+1)$  \\ \bottomrule
\end{tabular}
\end{table}

\begin{table}[!htb]
\caption{Critical-time (s) - Upper Saturation (diff. simulation vs algorithm)}
\label{tab_CT_by_Alg_Uppersat}
\centering
\begin{tabular}{c|ccccc}
\toprule
	\diagbox{$\gamma_1$}{$\gamma_2$}	   & $0.55$      & $0.575$ & $0.60$ & $0.625$ & $0.65$  \\ \midrule
  $0.65$  & $9$  & $11$         & $10$ 			  & $8$ 			  & $7$  \\
  $0.675$ & $7$  & $10$   			& $13$  			& $13$ $(+1)$ & $13$  \\ 
	$0.70$  & $6$  & $9$         	& $12$        & $12$ $(+1)$ & $13$   \\
  $0.725$ & $5$  & $8$   				& $11$  			& $12$        & $12$  \\ 
  $0.75$  & $4$  & $7$   				& $10$ $(+1)$ & $11$        & $11$  $(+1)$  \\ \bottomrule
\end{tabular}
\end{table}

			\subsubsection{Worst-case scenario}
			
			%$\rightarrow$ worst-case such as an adversary takes full control of the actuators
			%$\rightarrow$ help to find a good tradeoff with anomalies with opposite impact
			
		\begin{figure}[!htb]
      \centering
      \includegraphics[width=0.75\columnwidth]{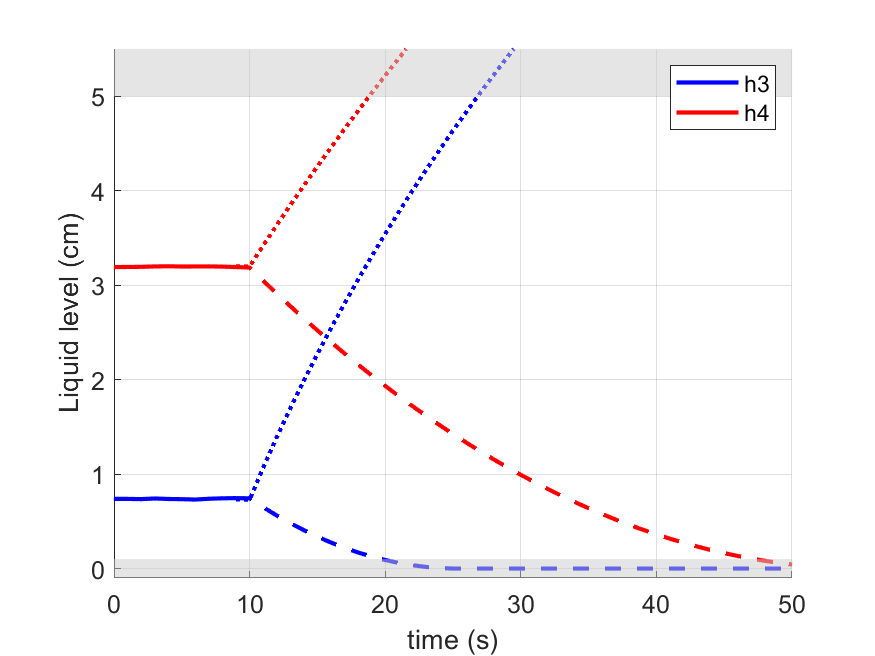}
      \caption{Impact on~$h_3$ and~$h_4$ of a DoS (dashed lines) and an attack by upper saturation (dotted lines) on the quadruple tank system originally at an operating point (solid lines). Grey patches refer to unsafe zones.}
      \label{fig_Simu_DoS_sat_attack}
   \end{figure}

				In this scenario, the class of anomalies leading to any admissible input $(\Gamma_u,\Gamma_a,\mathcal{A}) = (0,I,\mathcal{U})$ is considered.
				%This worst-case scenario may typically happen when an adversary takes full control of the actuators, either through the communication network or by a direct attack on the controller.
				This worst-case scenario may typically happen when an adversary takes full control of the inputs, either through the communication network or by an attack on the controller.
				In that case, a difficulty is to design a defense strategy for anomalies with opposite impact (see for instance Fig.~\ref{fig_Simu_DoS_sat_attack}). 
				Indeed, one may observe for instance in Table~\ref{tab_CT_by_Alg_DoS} and Table~\ref{tab_CT_by_Alg_Uppersat} that the couple $(\gamma_1,\gamma_2) = (0.70, 0.55)$ leads to the highest critical-time ($14~\text{s}$) for a DoS anomaly but low critical-time ($6~\text{s}$) when considering an attack-by-upper saturation.
				While an intuitive solution may be found on this example and these two particular anomalies, a suitable trade-off is expected to be more challenging to be made for numerous and diverse anomalies. 
				On the other hand, the computation of the critical-time for all admissible anomalies naturally incorporates this tradeoff, without needing to simulate each element of~$\mathcal{U}$.
				
		The computed under-estimation of the critical-time in this scenario is displayed in Table~\ref{tab_CT_by_Alg}.
		In addition, the computational time needed by Algorithm~\ref{Algo_CT_computation} at each iteration is plotted on Fig.~\ref{fig_average_time_by_iter_worstcase}. As expected, the computational time increases moderately with the time horizon~$k$ and remains tractalbe..
		The execution time of Algorithm~\ref{Algo_CT_computation} is $7.3~s$ on average.
				
			%%The computed under-estimation of the critical-time is displayed in Table~\ref{tab_CT_by_Alg}.
			%%One may observe that the choice of~$(\gamma_1,\gamma_2)$ has an important impact on the critical-time, varying from~$4 \text{s}$ to~$12 \text{s}$.
			%It can be observed in Table~\ref{tab_CT_by_Alg} that the choice of~$(\gamma_1,\gamma_2)$ has an important impact on the critical-time, varying from~$4 \text{s}$ to~$12 \text{s}$.
			%Thus, the interest of the critical-time metric is considered to be at least two-fold. 
			%First, it provides a means to evaluate the maximum time allocated to defense mechanisms to detect and mitigate anomalies. 
			%%This is expected to be especially useful when one has to design a defense strategy.
			%%Second, the critical-time criteria may be used as a prevention tool, as its maximization leads to extra-time allowed to defense mechanisms.
			%Second, it may be used as a prevention tool, as its maximization leads to extra-time allowed to defense mechanisms.
			%%Indeed, maximizing the critical-time may be viewed as a rational security allocation approach~\cite{CST:19}, where the aim is to tune degrees of freedom in the controlled system design in order to make the security problem easier, rather than developing complex defense strategies. 			
			%Indeed, maximizing the critical-time may be viewed as a rational security allocation approach~\cite{CST:19}, where the aim is to tune degrees of freedom in the controlled system in order to make the security problem easier rather than developing complex defense strategies. 			

\begin{table}[!htb]
%\caption{Critical-time estimation by Algorithm~\ref{Algo_CT_computation} for different values of~$(\gamma_1,\gamma_2)$} 
%\caption{Critical-time estimation by Algorithm~\ref{Algo_CT_computation} for~$(\gamma_1,\gamma_2)$} 
\caption{Critical-time (s) - Worst-case scenario} 
\label{tab_CT_by_Alg}
\centering
\begin{tabular}{c|ccccc}
\toprule
	%\diagbox[width=3em]{$\gamma_1$}{$\gamma_2$}	& $C_{0_1}$ & $L_{m_1}$ & $C_{m_1}$ & $C_{0_2}$ & $L_{m_2}$ & $C_{m_2}$  \\[0.25em]
	\diagbox{$\gamma_1$}{$\gamma_2$}	& $0.55$ & $0.575$ & $0.60$ & $0.625$ & $0.65$  \\ \midrule
  $0.65$  & $9$ & $11$ & $10$ & $8$  & $6$  \\
  $0.675$ & $7$ & $10$ & $11$ & $10$ & $8$  \\ 
	$0.70$  & $6$ & $9$  & $12$ & $11$ & $9$  \\
  $0.725$ & $5$ & $8$  & $11$ & $12$ & $10$ \\ 
  $0.75$  & $4$ & $7$  & $10$ & $11$ & $11$   \\ \bottomrule
\end{tabular}
\end{table}

\begin{figure}[!htb]
  \centering
   \includegraphics[width=0.75\columnwidth]{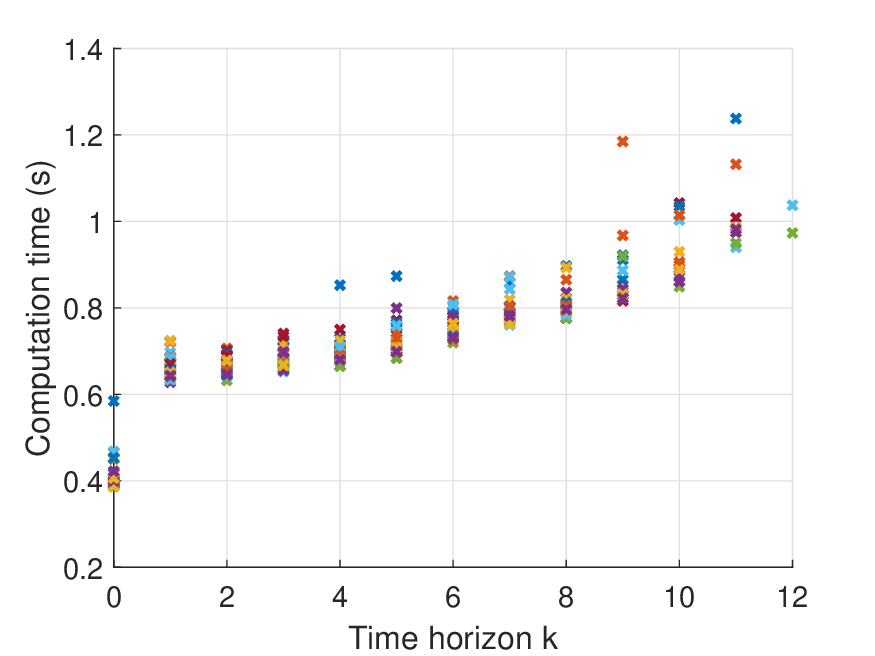}
   \caption{Computational time per iteration in the worst-case scenario}
   \label{fig_average_time_by_iter_worstcase}
\end{figure}

	\subsubsection{Channel-dependent scenario}
		
		The classes of anomalies leading to any admissible input for a single input channel is now considered.
		%In the quadruple tank system, these are modeled by $(\Gamma_{u_1}, \Gamma_{a_1}, \mathcal{A}_1) = \left(\begin{bmatrix} 0 & 0 \\ 0 & 1 \end{bmatrix}, \begin{bmatrix} 1 \\ 0 \end{bmatrix}, \mathcal{U}_1\right)$ for the first channel~$C_1$ and by $(\Gamma_{u_2}, \Gamma_{a_2}, \mathcal{A}_2) = \left(\begin{bmatrix} 1 & 0 \\ 0 & 0 \end{bmatrix}, \begin{bmatrix} 0 \\ 1 \end{bmatrix}, \mathcal{U}_2\right)$ for the second channel~$C_2$, where
		These are modeled by 
		\begin{equation*}
					(\Gamma_{u_1}, \Gamma_{a_1}, \mathcal{A}_1) = \left(\begin{bmatrix} 0 & 0 \\ 0 & 1 \end{bmatrix}, \begin{bmatrix} 1 \\ 0 \end{bmatrix}, \mathcal{U}_1\right)
		\end{equation*}
		for the first channel~$C_1$ and by 
		\begin{equation*}
				(\Gamma_{u_2}, \Gamma_{a_2}, \mathcal{A}_2) = \left(\begin{bmatrix} 1 & 0 \\ 0 & 0 \end{bmatrix}, \begin{bmatrix} 0 \\ 1 \end{bmatrix}, \mathcal{U}_2\right)
		\end{equation*}
		for the second channel~$C_2$, where
		\begin{equation*}
				%\begin{aligned}
					%\mathcal{U}_1 &:= \left\{ u \in \mathbb{R}, \ \begin{bmatrix} u \\ 0 \end{bmatrix} \in \mathcal{U} \right\} \\
					%\mathcal{U}_2 &:= \left\{ u \in \mathbb{R}, \ \begin{bmatrix} 0 \\ u \end{bmatrix} \in \mathcal{U} \right\}
				%\end{aligned}
					\mathcal{U}_1 := \left\{ u \in \mathbb{R}, \ \begin{bmatrix} u \\ 0 \end{bmatrix} \in \mathcal{U} \right\} 
					\quad
					\mathcal{U}_2 := \left\{ u \in \mathbb{R}, \ \begin{bmatrix} 0 \\ u \end{bmatrix} \in \mathcal{U} \right\}				
		\end{equation*}
		%%The estimated critical-time associated with channel-dependent sharp anomalies appears especially useful from a risk analysis perspective as it allows to identify which channel should be protected as a priority for defender with limited resources~\cite{TSSJ:15}.
		%%For instance, channel~C2 appears here more vulnerable than C1 as it generally leads to lower critical-time (Table~\ref{tab_CT_by_Alg_CD}).
		%%Moreover, 
		%%
		%The estimated critical-time associated with channel-dependent sharp anomalies is especially useful for the design of the channel-isolation step, required for the selection of a suitable mitigation strategy, in a detection-and-isolation scheme~\cite{HKKS:10}.
		%%,PDB:13}.
		%%Moreover, it allows to identify which channel should be protected as a priority for defender with limited resources in a risk management approach~\cite{TSSJ:15}.
		%Moreover, from a risk analysis perspective, it allows to identify which channel should be protected as a priority for defender with limited resources~\cite{TSSJ:15}.
		%For instance, channel~C2 appears here more vulnerable than C1 as it generally leads to lower critical-time (Table~\ref{tab_CT_by_Alg_CD}).
		
		The estimated critical-time by Algorithm~\ref{Algo_CT_computation} is displayed in Fig.~\ref{tab_CT_by_Alg_CD}.
			
			%$\rightarrow$ especially important for detection and isolation	
			%
			%$\rightarrow$ + allow to check which channel is the most fragile/less unsecure	

\begin{table}[!htb]
\caption{Critical-time (s) - Channel-dependent scenario C1 $|$ C2}
\label{tab_CT_by_Alg_CD}
\centering
\begin{tabular}{c|ccccc}
\toprule
	\diagbox{$\gamma_1$}{$\gamma_2$}	& $0.55$ & $0.575$ & $0.60$ & $0.625$ & $0.65$  \\ \midrule
  $0.65$  & $14$ $|$ $9$ & \hspace{0.18em} $11$ $|$ $12$ & $10$ $|$ $10$ & $8$ $|$ $8$  & \hspace{0.18em} $7$ $|$ $6$  \\
  $0.675$ & $14$ $|$ $7$ & \hspace{0.18em} $14$ $|$ $10$ & $13$ $|$ $12$ & $12$ $|$ $10$ & $11$ $|$ $8$  \\ 
	$0.70$  & $13$ $|$ $6$ & $13$ $|$ $9$  & $13$ $|$ $12$ & $12$ $|$ $11$ & $12$ $|$ $9$  \\
  $0.725$ & $10$ $|$ $5$ & $12$ $|$ $8$  & $13$ $|$ $11$ & $13$ $|$ $12$ & \hspace{0.18em}  $12$ $|$ $10$ \\ 
  $0.75$  & \hspace{0.18em} $7$ $|$ $4$ & $10$ $|$ $7$  & $11$ $|$ $10$ & $12$ $|$ $14$ & \hspace{0.18em} $12$ $|$ $11$   \\ \bottomrule
\end{tabular}
\end{table}

\subsubsection{Interpretation}

			In Table~\ref{tab_CT_by_Alg_DoS}-\ref{tab_CT_by_Alg_CD}, it can be observed that the choice of~$(\gamma_1,\gamma_2)$ has an important impact on the critical-time, varying from~$4 \text{s}$ to~$12 \text{s}$ in Table~\ref{tab_CT_by_Alg} for instance.
			The interest of the critical-time metric is then considered to be at least two-fold. 
			First, it provides a means to evaluate the maximum time allocated to defense mechanisms to detect and mitigate anomalies. 
			%This is expected to be especially useful when one has to design a defense strategy.
			%Second, the critical-time criteria may be used as a prevention tool, as its maximization leads to extra-time allowed to defense mechanisms.
			Second, it may be used as a prevention tool, as its maximization leads to extra-time allowed to defense mechanisms.
			%Indeed, maximizing the critical-time may be viewed as a rational security allocation approach~\cite{CST:19}, where the aim is to tune degrees of freedom in the controlled system design in order to make the security problem easier, rather than developing complex defense strategies. 			
			Indeed, maximizing the critical-time may be viewed as a rational security allocation approach~\cite{CST:19}, where the aim is to tune degrees of freedom in the controlled system in order to make the security problem easier rather than developing complex defense strategies. 	
			
			Finally, the estimated critical-time associated with channel-dependent sharp anomalies is especially useful for the design of the channel-isolation step, required for the selection of a suitable mitigation strategy, in a detection-and-isolation scheme~\cite{HKKS:10}.
		%,PDB:13}.
		%Moreover, it allows to identify which channel should be protected as a priority for defender with limited resources in a risk management approach~\cite{TSSJ:15}.
		Moreover, from a risk analysis perspective, it allows to identify which channel should be protected as a priority for defender with limited resources~\cite{TSSJ:15}.
		For instance, channel~C2 appears here more vulnerable than C1 as it generally leads to lower critical-time (Table~\ref{tab_CT_by_Alg_CD}).

\section{Conclusion} \label{sec_conclusion}

	In this paper, the critical-time metric was investigated for controlled systems subject to sharp input anomalies.  
	This metric has several advantages for off-line risk analysis as it allows for quantifying the time-window available to defense mechanisms during which safe operation can be ensured.	
	Using the QC framework, an iterative LMI-based algorithm has been established to compute the critical-time for an uncertain linear discrete-time model.
	Several classes of sharp input anomalies can be considered, depending on the input channel and the set of values the abnormal signal can take.
	%In addition, the potential of the critical-time for defense design and as a resilience-increasing mechanism has been illustrated on the quadruple tank benchmark through the analysis of a representative set of distinct sharp anomalies scenarios.
	In addition, the potential of the critical-time for risk analysis and treatment has been illustrated on the quadruple tank benchmark through the analysis of a representative set of distinct sharp anomalies scenarios.

	%Future work includes extension to LPV models, the investigation of the critical-time for stealthy attacks, and the integration of specific communication models.
	Future work includes extension to LPV models, the investigation of the critical-time for stealthy attacks, and the integration of specific communication models.
	%This work paves the way for further investigations.  
	%This includes critical-time computation for stealthy attacks, extension to LPV models, and the integration of control and communication models. This is a topic of our current research.
	
	%+ include other safety sets ?
	
	%Perspectives (papier)
	%\begin{itemize}
		%\item QC Inner approximation of sets with non-exact QC representations
		%\item Selection of initial state ?
		%\item Problem 1 adapted for sharp anomalies --> Critical-time computation for stealthy attack
		%\item As defense mechanism (detection filter etc) are usually different in the literature and has a different operating delay (detection, reconfiguration delays) depending if the anomaly is sharp or more sophisticated (stealthy, covert attacks).
		%\item Pour l'instant un peu grossier car same criteria for sharp and stealthy anomalies ==> future work to extend this criteria for stealthy anomalies or covert attack.
		%\item Extension to the case where unsafe set is convex and the safety set is non-convex (collision avoidance) ? To the case where the safety set is the union of convex set ?
	%\end{itemize}
	
	%Extensions (pour moi)
	%\begin{itemize}
		%\item (Rational-SOS) Continuous-time systems
		%\item Consider non-linear systems with QCs ?
	%\end{itemize}



  %\bibliographystyle{elsarticle-harv} 
  \bibliographystyle{elsarticle-num} 
  \bibliography{BIB_LCSS_2022}

%%% else use the following coding to input the bibitems directly in the
%%% TeX file.
%
%\begin{thebibliography}{00}
%
%%% \bibitem[Author(year)]{label}
%%% Text of bibliographic item
%\end{thebibliography}
\end{document}